\documentclass[a4paper,UKenglish,cleveref, autoref]{lipics-v2019}

\title{Twin-width II: small classes}
\titlerunning{Twin-width II: small classes}

\author{\'{E}douard Bonnet}{Univ Lyon, CNRS, ENS de Lyon, Université Claude Bernard Lyon 1, LIP UMR5668, France}{edouard.bonnet@ens-lyon.fr}{https://orcid.org/0000-0002-1653-5822}{}
\author{Colin Geniet}{ENS Paris-Saclay, France}{colin.geniet@ens-paris-saclay.fr}{}{}
\author{Eun Jung Kim}{Universit\'{e} Paris-Dauphine, PSL University, CNRS UMR7243, LAMSADE, Paris, France}{eun-jung.kim@dauphine.fr}{https://orcid.org/0000-0002-6824-0516}{}
\author{St\'{e}phan Thomass\'{e}}{Univ Lyon, CNRS, ENS de Lyon, Universit\'{e} Claude Bernard Lyon 1, LIP UMR5668, France}{stephan.thomasse@ens-lyon.fr}{}{}
\author{R\'{e}mi Watrigant}{Univ Lyon, CNRS, ENS de Lyon, Universit\'{e} Claude Bernard Lyon 1, LIP UMR5668, France}{remi.watrigant@univ-lyon1.fr}{https://orcid.org/0000-0002-6243-5910}{}

\authorrunning{\'E. Bonnet, C. Geniet, E. J. Kim, S. Thomassé, R. Watrigant}

\Copyright{Édouard Bonnet, Colin Geniet, Eun Jung Kim, Stéphan Thomassé, Rémi Watrigant}



\category{}

\relatedversion{}

\supplement{}

\acknowledgements{}

\nolinenumbers 

\hideLIPIcs  

\EventEditors{John Q. Open and Joan R. Access}
\EventNoEds{2}
\EventLongTitle{42nd Conference on Very Important Topics (CVIT 2016)}
\EventShortTitle{CVIT 2016}
\EventAcronym{CVIT}
\EventYear{2016}
\EventDate{December 24--27, 2016}
\EventLocation{Little Whinging, United Kingdom}
\EventLogo{}
\SeriesVolume{42}
\ArticleNo{23}

\usepackage[utf8]{inputenc}  

\usepackage[T1]{fontenc}
\usepackage{lmodern}

\usepackage[colorinlistoftodos,bordercolor=orange,backgroundcolor=orange!20,linecolor=orange,textsize=normalsize]{todonotes}

\usepackage{amsmath}  
\usepackage{amssymb}     
\usepackage{bbm}

\usepackage{thm-restate}
\usepackage{thmtools}

\usepackage{booktabs}
\usepackage{bm}
\usepackage{hyperref}
\usepackage{mathrsfs}
\usepackage{mathtools} 
\usepackage{paralist}
\usepackage{fixmath}
\usepackage{interval}

\crefformat{equation}{\textup{#2(#1)#3}}
\crefrangeformat{equation}{\textup{#3(#1)#4--#5(#2)#6}}
\crefmultiformat{equation}{\textup{#2(#1)#3}}{ and \textup{#2(#1)#3}}
{, \textup{#2(#1)#3}}{, and \textup{#2(#1)#3}}
\crefrangemultiformat{equation}{\textup{#3(#1)#4--#5(#2)#6}}%
{ and \textup{#3(#1)#4--#5(#2)#6}}{, \textup{#3(#1)#4--#5(#2)#6}}{, and \textup{#3(#1)#4--#5(#2)#6}}

\Crefformat{equation}{#2Equation~\textup{(#1)}#3}
\Crefrangeformat{equation}{Equations~\textup{#3(#1)#4--#5(#2)#6}}
\Crefmultiformat{equation}{Equations~\textup{#2(#1)#3}}{ and \textup{#2(#1)#3}}
{, \textup{#2(#1)#3}}{, and \textup{#2(#1)#3}}
\Crefrangemultiformat{equation}{Equations~\textup{#3(#1)#4--#5(#2)#6}}%
{ and \textup{#3(#1)#4--#5(#2)#6}}{, \textup{#3(#1)#4--#5(#2)#6}}{, and \textup{#3(#1)#4--#5(#2)#6}}

\crefmultiformat{claim}{claims~#2#1#3}{ and~#2#1#3}{, #2#1#3}{, and~#2#1#3}

\usepackage{xspace}
\usepackage{tikz}

\usetikzlibrary{fit}
\usetikzlibrary{arrows}
\usetikzlibrary{patterns}
\usetikzlibrary{calc}
\usetikzlibrary{shapes}
\usetikzlibrary{positioning}
\usetikzlibrary{math}
\usetikzlibrary{shapes.symbols}

\usepackage[scr=boondox,scrscaled=1.05]{mathalfa}

\makeatletter
\def\Ddots{\mathinner{\mkern1mu\raise\p@
\vbox{\kern7\p@\hbox{.}}\mkern2mu
\raise4\p@\hbox{.}\mkern2mu\raise7\p@\hbox{.}\mkern1mu}}
\makeatother

\renewcommand{\geq}{\geqslant}
\renewcommand{\leq}{\leqslant}
\renewcommand{\preceq}{\preccurlyeq}

\renewcommand{\le}{\leq}
\renewcommand{\ge}{\geq}

\newcommand{\card}[1]{\left|{#1}\right|}

\newtheorem{conjecture}[theorem]{Conjecture}

\theoremstyle{definition}

\newcommand{\tww}{\text{tww}}
\newcommand{\sub}{\text{Sub}}

\newcommand{\cay}{\text{Cay}}
\newcommand{\OLCay}{\text{OLCay}}
\newcommand{\OLF}{\text{OLF}}

\begin{document}

\maketitle

\begin{abstract}
  The recently introduced \emph{twin-width} of a graph $G$ is the minimum integer $d$ such that $G$ has a \emph{$d$-contraction sequence}, that is, a sequence of $\card{V(G)}-1$ iterated vertex identifications for which the overall maximum number of red edges incident to a single vertex is at most~$d$, where a red edge appears between two sets of identified vertices if they are not homogeneous in $G$ (not fully adjacent nor fully non-adjacent). 
  We show that if a graph admits a $d$-contraction sequence, then it also has a linear-arity tree of $f(d)$-contractions, for some function $f$.
  Informally if we accept to worsen the twin-width bound, we can choose the next contraction from a set of $\Theta(\card{V(G)})$ pairwise disjoint pairs of vertices. 
  This has two main consequences.
  First it permits to show that every bounded twin-width class is \emph{small}, i.e., has at most $n!c^n$ graphs labeled by $[n]$, for some constant $c$.
  This unifies and extends the same result for bounded treewidth graphs [Beineke and Pippert, JCT '69], proper subclasses of permutations graphs [Marcus and Tardos, JCTA '04], and proper minor-free classes [Norine et al., JCTB '06].
  It implies in turn that bounded-degree graphs, interval graphs, and unit disk graphs have unbounded twin-width.
  The second consequence is an $O(\log n)$-adjacency labeling scheme for bounded twin-width graphs, confirming several cases of the implicit graph conjecture.

  We then explore the \emph{small conjecture} that, conversely, every small hereditary class has bounded twin-width.
  The conjecture passes many tests.
  Inspired by sorting networks of logarithmic depth, we show that $\log_{\Theta(\log \log d)}n$-subdivisions of $K_n$ (a small class when $d$ is constant) have twin-width at most~$d$.
  We obtain a rather sharp converse with a surprisingly direct proof: the $\log_{d+1}n$-subdivision of $K_n$ has twin-width at least~$d$.
  Secondly graphs with bounded stack or queue number (also small classes) have bounded twin-width.
  These sparse classes are surprisingly rich since they contain certain (small) classes of expanders.  
  Thirdly we show that cubic expanders obtained by iterated random 2-lifts from $K_4$~[Bilu and Linial, Combinatorica '06] also have bounded twin-width.
  These graphs are related to so-called separable permutations and also form a small class.
  We suggest a promising connection between the small conjecture and group theory.

  Finally we define a robust notion of sparse twin-width.
  We show that for a hereditary class $\mathcal C$ of bounded twin-width the five following conditions are equivalent: every graph in $\mathcal C$ (1) is $K_{t,t}$-free for some fixed~$t$, (2) has an adjacency matrix without a $d$-by-$d$ division with a 1 entry in each $d^2$ cells for some fixed~$d$, (3) has at most linearly many edges, (4) the subgraph closure of~$\mathcal C$ has bounded twin-width, and (5)~$\mathcal C$~has bounded expansion.
  We discuss how sparse classes with similar behavior with respect to clique subdivisions compare to bounded sparse twin-width. 
\end{abstract}
\maketitle

\section{Introduction}\label{sec:intro}

We continue to develop the theory of twin-width, a novel graph and matrix invariant introduced in the first paper of the series~\cite{twin-width1}.
We start with a bird's eye view of our results.
The exact definitions of some objects and concepts will be deferred to the next section, but this introduction can be read by taking them as black boxes.
Furthermore \cref{sec:prelim} includes a summary of the first paper, so that the current paper is self-contained. 

A~trigraph is a graph with two disjoint edge sets: black edges (regular edges) and red edges (error edges).
The graph induced by the red edges (resp.~black edges) is called the red graph (resp.~black graph). 
A~$d$-trigraph has a red graph with maximum degree at most~$d$.
A~contraction in a trigraph identifies two (non-necessarily adjacent) vertices, and puts black edges towards shared neighbors in the black graph, and red edges towards the other (non-necessarily shared) neighbors (see~\cref{fig:contraction}).
A~$d$-contraction sequence, or $d$-sequence, of an $n$-vertex graph $G$ is a sequence of $d$-trigraphs $G=G_n, G_{n-1}, \ldots, G_2, G_1$ such that $G_i$ is obtained by performing a single contraction in $G_{i+1}$.
In particular $G_1$ is the one-vertex graph $K_1$.
The twin-width of $G$ is the minimum $d$ such that it admits a $d$-sequence.

A~contraction sequence of $G$ may be seen as a path with at the left end, $G$, at the right end, $K_1$, and the current trigraph gets smaller and smaller when we walk from left to right.
We show that this path can be made a tree of large arity.
Now $G$ is at the root of the tree, all the leaves contain the graph $K_1$, and every child is obtained by performing a single contraction in the parent node.
A~$d$-contraction tree is such a tree with a $d$-trigraph at every node. 
More precisely, we show that if a graph $G$ has a $d$-contraction sequence, then it has a $D_d$-contraction tree with linear arity.
By linear arity, we mean that every non-leaf node $H$ has $\Theta(\card{V(H)})$ distinct children.

Denoting the class of graphs with twin-width at most $d$ by $\mathcal C_d$, the first consequence is that the number of graphs in $\mathcal C_d$ on the vertex set $[n]$ is at most $n!f(d)^n$.
Intuitively the large-arity tree tells us that many $n-1$-vertex graphs of $\mathcal C_d$ can be obtained from the same $n$-vertex graph of $\mathcal C_d$.
By inverting the process, there are not so many \emph{distinct} $n$-vertex graphs in $\mathcal C_d$, obtained by splitting a vertex in $n-1$-vertex graphs of $\mathcal C_d$.
This crucial fact makes the inductive proof works.
Our result generalizes several similar theorems in enumerative combinatorics.

The first one is an over 50-year old result that bounded treewidth graphs on vertex set~$[n]$ have a similar growth in $n!c^n$~\cite{Beineke69}.
Graph classes with such a growth are called small.
The second one is comparatively much more recent, it is the celebrated answer to the Stanley-Wilf conjecture, now the Marcus-Tardos theorem.
Marcus and Tardos~\cite{MarcusT04} showed that there are at most $c_\sigma^n$ permutations over $[n]$ avoiding a fixed permutation pattern $\sigma$.
In other words, every proper subclass of permutations (where a class of permutations is closed under taking subpermutations) has at most single-exponential growth, much below $n!$, the growth of the full class.
Expressed in the language of graph classes, proper subclasses of permutation graphs are small.
The third one, due to Norine et al.~\cite{Norine06}, is that the number of graphs on vertex set $[n]$ not containing a fixed minor $H$ is at most $n!c_H^n$.
Thus proper minor-closed classes are small.

We previously showed \cite{twin-width1} that bounded treewidth (even rank-width) graphs, proper subclasses of permutation graphs, and proper minor-closed classes have bounded twin-width.
Thus the fact that bounded twin-width classes are small unifies and extends all the above-mentioned theorems.
We then explore the converse statement.
Could it be that every small hereditary class has bounded twin-width?
We do not answer this question, dubbed the small conjecture, but instead we give some evidences it may be true.
This comes in the form of showing that many potential counterexamples, that is, seemingly complex small hereditary classes, actually have bounded twin-width. 
If the conjecture is true, it gives a universal explanation for the single-exponential growth (up to isomorphism) of combinatorial classes: Translate the objects into graphs or matrices, a bound or lack thereof in the twin-width of the class decides the existence of such a bound in the growth. 

Another by-product of the contraction tree is that we can always contract in parallel a~linear number of disjoint pairs of vertices.
This gives rise to so-called parallel $d$-sequences of logarithmic length.
This will be instrumental in showing that bounded twin-width classes admit an $O(\log n)$-adjacency labeling scheme.
This verifies a variety of particular cases of the implicit graph conjecture which posits that such labeling schemes exist for every factorial hereditary class, i.e., hereditary class with growth $n!^{O(1)}$. 

Finally we show that five different ways of restricting twin-width to sparse classes actually lead to the same notion.
For example, bounded sparse twin-width classes can be equivalently defined as hereditary classes with bounded twin-width that are $K_{t,t}$-free or where every graph has at most linearly many edges.
A first but challenging step towards the small conjecture is to show that small sparse classes have bounded (sparse) twin-width.
For instance, do classes with polynomial expansion have bounded twin-width?
We discuss (possible) containments and strict containments of established sparse classes with respect to bounded sparse twin-width.

\section{Preliminaries and outline}\label{sec:prelim}

In this section we recall the relevant notations and definitions, summarize the important bits of the first paper, and outline our new results.

\subsection{Notations and definitions}

We denote by $[i,j]$ the set of integers $\{i,i+1,\ldots, j-1, j\}$, and by $[i]$ the set of integers $[1,i]$.
If $\mathcal X$ is a set of sets, we denote by $\cup \mathcal X$ their union.
Unless stated otherwise, all graphs are assumed undirected and simple, that is, they do not have parallel edges or self-loops.
We denote by $V(G)$ and $E(G)$, the set of vertices and edges, respectively, of a graph $G$. 
For $S \subseteq V(G)$, we denote the \emph{open neighborhood} (or simply \emph{neighborhood}) of $S$ by $N_G(S)$, i.e., the set of neighbors of $S$ deprived of $S$, and the \emph{closed neighborhood} of $S$ by $N_G[S]$, i.e., the set $N_G(S) \cup S$.
We simplify $N_G(\{v\})$ into $N_G(v)$, and $N_G[\{v\}]$ into $N_G[v]$.
We denote by $G[S]$ the subgraph of $G$ induced by $S$, and $G - S := G[V(G) \setminus S]$.
For two disjoint sets $A, B \subseteq V(G)$, $E(A,B)$ denotes the set of edges in $E(G)$ with one endpoint in $A$ and the other one in $B$.
Two distinct vertices $u, v$ such that $N(u) = N(v)$ are called \emph{false twins}, and \emph{true twins} if $N[u] = N[v]$.
Two vertices are \emph{twins} if they are false twins or true twins.
For two vertices $u, v \in V(G)$, the \emph{distance $d_G(u,v)$} is the number of edges in a shortest path from $u$ to $v$, and $\infty$ if $u$ and $v$ are in two distinct connected components of $G$.
In all the notations with a graph subscript, we may omit it if the graph is clear from the context.

A~\emph{graph class} is a family of graphs closed under isomorphism (i.e., under renaming the vertices).
Since we will be interested in the ``size'' of a class, we will further impose that the vertex set of $n$-vertex graphs is precisely\footnote{If it is sometimes more convenient to use a different vertex set for the class definition, this will implicitly come with a canonical mapping from this vertex set to $[n]$.}~$[n]$.
With that requirement the number of $n$-vertex graphs in a class $\mathcal C$ is a well-defined (finite) number.
Observe that every single $n$-vertex graph in a class $\mathcal C$ implies that at least $n!$ graphs are in $\mathcal C$, namely all its relabelings.
A graph class is said~\emph{hereditary} if it is closed under taking induced subgraphs.
It is said~\emph{monotone} or \emph{subgraph-closed} if it is even closed under taking subgraphs.

A graph is \emph{$H$-free} if it does not contain $H$ as an induced subgraph.
However we make an exception for $H = K_{t,t}$.
A~$K_{t,t}$-free graph is a graph with no biclique $K_{t,t}$ \emph{as a subgraph}.
A~class is \emph{$H$-free} if all its graphs are $H$-free.
When $t$ is not yet defined, we may say that a class $\mathcal C$ is $K_t$-free (resp.~$K_{t,t}$-free) to mean that there exists a finite integer $t$ such that $\mathcal C$ is $K_t$-free (resp.~$K_{t,t}$-free). 

We denote by $\Delta(G)$ the maximum degree of a vertex in $G$, and $\Delta(\mathcal C) := \sup_{G \in \mathcal C} \Delta(G)$.
A~class $\mathcal C$ has \emph{bounded degree} if $\Delta(\mathcal C) < \infty$.
More generally, for any graph invariant~$\iota$, we say that \emph{$\mathcal C$ has bounded~$\iota$} if $\iota(\mathcal C) := \sup_{G \in \mathcal C} \iota(G) < \infty$.
The strong product $G \boxtimes H$ of two graphs $G$ and $H$ has vertex set $V(G) \times V(H)$ and $(u,v)(u',v') \in E(G \boxtimes H)$ if and only if [$u = u'$ or $uu' \in E(G)$] and [$v=v'$ or $vv' \in E(H)$].
We denote by $\mathcal G \boxtimes \mathcal H$ the class $\{G \boxtimes H$ $|$ $G \in \mathcal G, H \in \mathcal H\}$, where $\mathcal G$ and $\mathcal H$ are two sets of graphs.
Given a class $\mathcal C$, we denote by $\sub(\mathcal C)$ the class of all subgraphs of members of $\mathcal C$.
The class $\sub(\mathcal C)$ is by definition subgraph-closed, and is called the \emph{subgraph closure} of $\mathcal C$. 
Similarly the \emph{hereditary closure} of a class $\mathcal C$ consists of all the induced subgraphs of members of $\mathcal C$, and is hereditary by design.

An \emph{edge contraction} of two adjacent vertices $u, v$ consists of merging $u$ and $v$ into a single vertex adjacent to $N(\{u,v\})$ (and deleting $u$ and $v$).
A graph $H$ is a \emph{minor} of a graph $G$ if $H$ can be obtained from $G$ by a sequence of vertex and edge deletions, and edge contractions.
Equivalently a minor $H$ with vertex set say, $\{v_1, \ldots, v_{V(H)}\}$, of $G$ can be defined as a vertex partition $B_1, \ldots, B_{|V(H)|}$ of a subgraph of $G$, such that every $G[B_i]$ is connected and $E_G(B_i,B_j) \neq \emptyset$ whenever $v_iv_j \in E(H)$.
Indeed after contracting each $B_i$ into a single vertex (which is possible since they induce connected subgraphs), $H$ appears as a subgraph.
The set $B_i$ is called the \emph{branch set} of $v_i \in V(H)$.  
A graph $G$ is said \emph{$H$-minor free} if $H$ is not a minor of $G$.
A class is said \emph{minor-closed} if every minor of a member of the class is in the class, and \emph{proper minor-closed} if further the class is \emph{not} the set of all graphs.

The \emph{radius $\text{rad}(G)$} of a graph $G$ is defined as $\min_{u \in V(G)} \max_{v \in V(G)} d_G(u,v)$.
The \emph{radius $\text{rad}_G(S)$} of a subset of vertices $S \subseteq V(G)$ is simply defined as $\text{rad}(G[S])$.
Note that two vertices can be further away in $G[S]$ than in $G$.
An \emph{$r$-shallow minor} $H$ of $G$ is a minor of~$G$ with branch sets $B_1, \ldots, B_{\card{V(H)}}$ satisfying $\text{rad}_G(B_i) \leqslant r$ for every $i \in [\card{V(H)}]$.
We denote that by $H \preccurlyeq_r G$.
In particular 0-shallow minors correspond to subgraphs.
The theory of graph sparsity pioneered by Ossona de Mendez and Nešetřil~\cite{sparsity} introduces the following invariants for a graph $G$ and a class $\mathcal C$:
$$\nabla_r(G) := \underset{H \preccurlyeq_r G}{\sup}~\frac{|E(H)|}{|V(H)|},~\text{and}~\nabla_r(\mathcal C) := \underset{G \in \mathcal C}{\sup}~\nabla_r(G).$$
Note that $\nabla_0(G)$ is tied to the maximum average degree of $G$.

A class $\mathcal C$ of graphs is said to have \emph{bounded expansion} if $\nabla_r(\mathcal C) < \infty$ for every $r \in \mathbb N$.
More generally $\mathcal C$ has \emph{expansion $f$} if $\nabla_r(\mathcal C) \leqslant f(r)$ for every $r \in \mathbb N$.
A class has \emph{polynomial expansion} if it has expansion $f$ for a polynomial function $f$.
Proper minor-closed classes even have constant expansion, i.e., expansion $f$ for a constant function $f$.

\subsection{Summary of the previous paper}

In the previous paper of the series~\cite{twin-width1}, we introduced a new graph and matrix invariant dubbed \emph{twin-width}, inspired by the work of Guillemot and Marx on permutations~\cite{Guillemot14}.
We proved that many classes such as, bounded rank-width graphs, proper minor-free classes, proper subclasses of permutation graphs, and posets with antichains of bounded size have bounded twin-width.
For all these classes, we showed how to find in polynomial-time a~so-called \emph{$d$-sequence}, witnessing that the twin-width is at most a constant~$d$.
Finally given a $d$-sequence of a binary structure $G$ on $n$ elements and a first-order (FO) formula $\varphi$ of quantifier-depth~$\ell$, we provided an FO model checking algorithm deciding $G \models \varphi$ in time $f(d,\ell) n$.

We start by recalling the definition of twin-width, and then we summarize the milestones of \cite{twin-width1} that will also be useful in the current paper.

\subsubsection{Trigraphs, contraction sequences, and twin-width of a graph} 

A \emph{trigraph $G$} has vertex set $V(G)$, (black) edge set $E(G)$, and red edge set $R(G)$ (the error edges), with $E(G)$ and $R(G)$ being disjoint.
The \emph{set of neighbors $N_G(v)$} of a vertex $v$ in a trigraph $G$ consists of all the vertices adjacent to $v$ by a black or red edge.
A $d$-trigraph is a trigraph $G$ such that the \emph{red graph} $(V(G),R(G))$ has degree at most~$d$.
In that case, we also say that the trigraph has \emph{red degree} at most~$d$.
In the context of trigraphs and twin-width, we will somewhat overload the term ``contraction''.
A \emph{contraction} or \emph{identification} in a trigraph~$G$ consists of merging two (non-necessarily adjacent) vertices $u$ and $v$ into a single vertex $w$, and updating the edges of $G$ in the following way.
Every vertex of the symmetric difference $N_G(u) \triangle N_G(v)$ is linked to $w$ by a red edge.
Every vertex $x$ of the intersection $N_G(u) \cap N_G(v)$ is linked to $w$ by a black edge if both $ux \in E(G)$ and $vx \in E(G)$, and by a red edge otherwise.
The rest of the edges (not incident to $u$ or $v$) remain unchanged.
We insist that the vertices $u$ and $v$ (together with the edges incident to these vertices) are removed from the trigraph. 
See \cref{fig:contraction} for an illustration.
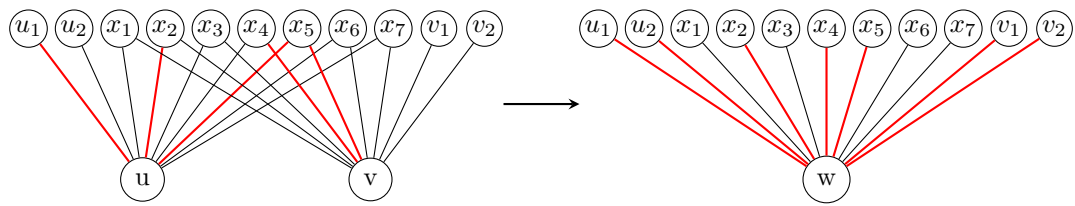
\begin{figure}
\begin{tikzpicture}
\def\v{2}
\def\t{6}
\def\s{0.6}

\draw[thick, -stealth] (3.25,\v /2) -- (4.25,\v/2) ;

\foreach \i/\j in {-5/u_1,-4/u_2,-3/x_1,-2/x_2,-1/x_3,0/x_4,1/x_5,2/x_6,3/x_7,4/v_1,5/v_2}{
  \node[draw,circle,inner sep=0.03cm] (n\i) at (\s * \i,\v) {$\j$} ; 
}

\node[draw,circle] (u) at (-1.5,0) {u} ;
\node[draw,circle] (v) at (1.5,0) {v} ;

\foreach \i in {-4,-3,...,-3,-1,0,2,3}{
  \draw (u) -- (n\i) ;
}
\foreach \i in {-5,-2,1}{
  \draw[thick, red] (u) -- (n\i) ;
}
\foreach \i in {-3,...,-1,2,3,...,5}{
  \draw (v) -- (n\i) ;
}
\foreach \i in {0,1}{
  \draw[thick, red] (v) -- (n\i) ;
}

\begin{scope}[xshift=7.5cm]
\node[draw,circle] (uv) at (0,0) {w} ;
\foreach \i/\j in {-5/u_1,-4/u_2,-3/x_1,-2/x_2,-1/x_3,0/x_4,1/x_5,2/x_6,3/x_7,4/v_1,5/v_2}{
  \node[draw,circle,inner sep=0.03cm] (m\i) at (\s * \i,\v) {$\j$} ; 
}

\foreach \i in {-3,-1,2,3}{
  \draw (uv) -- (m\i) ;
}
\foreach \i in {-5,-4,-2,0,1,4,5}{
  \draw[thick, red] (uv) -- (m\i) ;
}
\end{scope}
\end{tikzpicture}
\caption{Contraction of vertices $u$ and $v$, and how the edges of the trigraph are updated.}
\label{fig:contraction}
\end{figure}

A \emph{sequence of $d$-contractions} or \emph{$d$-sequence} is a sequence of $d$-trigraphs $G_n, G_{n-1}, \ldots, G_1$, where $G_n = G$, $G_1=K_1$ is the graph on a single vertex, and $G_{i-1}$ is obtained from $G_i$ by performing a single contraction of two (non-necessarily adjacent) vertices.
We observe that $G_i$ has precisely $i$ vertices, for every $i \in [n]$.
The twin-width of $G$, denoted by $\tww(G)$, is the minimum integer~$d$ such that $G$ admits a~$d$-sequence.
Going back to the overload of the word ``contraction'', in case we actually refer to the classical (edge) contraction, either we will use the term ``edge contraction'', or it will be clear from the context what is meant. 

\subsubsection{Partitions, divisions, red number, and twin-width of a matrix}

We now give two equivalent definitions for the twin-width of a matrix.
The first is based on a \emph{contraction sequence} where we progressively reduce the size of the matrix, and introduce error symbols~$r$.
The second (equivalent) definition is based on a \emph{coarsening sequence} where we progressively coarsen a partition of the rows and columns of the matrix.

The red number of a matrix is the maximum number of $r$~entries (error entry, the $r$ stands for red) in a single row or column.
Given an $n \times m$ matrix $M$ and two columns $C_i$ and $C_j$ (resp.~two rows $R_i$ and $R_j$), the \emph{contraction} of $C_i$ and $C_j$ (resp.~$R_i$ and $R_j$) is obtained by deleting $C_j$ (resp. $R_j$) and replacing every entry $m_{k,i}$ of $C_i$ (resp.~every entry $m_{i,k}$ of $R_i$) by $r$ whenever $m_{k,i} \neq m_{k,j}$ (resp.~$m_{i,k} \neq m_{j,k}$).
A \emph{$d$-contraction sequence} of matrix $M$ is sequence of successive contractions starting at $M$, ending at some $1 \times 1$ matrix, such that all matrices of the sequence have red number at most~$d$.
The \emph{twin-width} of a matrix $M$ is the smallest integer~$d$ such that $M$ admits a~$d$-contraction sequence.

We observe that when $M$ has twin-width at most~$d$, one can reorder its rows and columns such that every contraction is on two consecutive rows or two consecutive columns.
The reordered matrix is then called \emph{$d$-twin-ordered}.
The \emph{symmetric twin-width} of an $n \times n$ matrix $M$ is defined similarly, except that the contraction of rows $i$ and $j$ (resp.~columns $i$ and $j$) is immediately followed by the contraction of columns $i$ and $j$ (resp.~rows $i$ and $j$).
The symmetric twin-width of the adjacency matrix of a graph $G$ corresponds to the twin-width of $G$.

For the second definition of the twin-width of a matrix, we need to introduce a bit of vocabulary on partitions.
We say that a partition $\mathcal P$ of a set $S$ \emph{refines} a partition $\mathcal P'$ of $S$ if every part of $\mathcal P$ is contained in a part of $\mathcal P'$.
Conversely we say that $\mathcal P'$ is a coarsening of $\mathcal P$.
We will further assume that a coarsening is proper, that is, $\mathcal P'$ and $\mathcal P$ are distinct.
Given a partition $\mathcal P$ and two distinct parts $P, P'$ of $\mathcal P$, the \emph{elementary coarsening} of $P$ and $P'$ yields the coarsening $\mathcal P \setminus \{P,P'\} \cup \{P \cup P'\}$.
Informally an elementary coarsening is the merge of two parts.

Given an ${n \times m}$ matrix $M$, we call \emph{row-partition} (resp.~\emph{column-partition}) a partition of the rows (resp.~columns) of~$M$.
A \emph{$(k,\ell)$-partition}, or simply \emph{partition}, of a matrix $M$ is a pair $(\mathcal R=\{R_1,\dots ,R_k\}, \mathcal C=\{C_1,\dots ,C_\ell\})$ where $\mathcal R$ is a row-partition and $\mathcal C$ is a column-partition.
In a matrix partition $(\mathcal R,\mathcal C)$, each part $R \in \mathcal R$ is called a \emph{row-part}, and each part $C \in \mathcal C$ is called a \emph{column-part}.
An \emph{elementary coarsening} of a partition $(\mathcal R,\mathcal C)$ of a matrix $M$ is obtained by performing one elementary coarsening in $\mathcal R$ or in $\mathcal C$. 
We distinguish two canonical partitions of an $n \times m$ matrix $M$: the \emph{finest partition} where $(\mathcal R,\mathcal C)$ have size $n$ and $m$, respectively, and the \emph{coarsest partition} where $|\mathcal R|=|\mathcal C|=1$.

A \emph{coarsening sequence} of an $n \times m$ matrix $M$ is a sequence of partitions $(\mathcal R^1,\mathcal C^1),$ $\ldots ,(\mathcal R^{n+m-1},\mathcal C^{n+m-1})$ where
\begin{compactitem}
\item $({\mathcal R}^1,{\mathcal C}^1)$ is the finest partition,
\item $({\mathcal R}^{n+m-1},{\mathcal C}^{n+m-1})$ is the coarsest partition, and
\item for every $i \in [n+m-2]$, $({\mathcal R}^{i+1},{\mathcal C}^{i+1})$ is an elementary coarsening of $({\mathcal R}^i,{\mathcal C}^i)$.
\end{compactitem}

Given a subset $R$ of rows and a subset $C$ of columns in a matrix $M$, the \emph{zone $R \cap C$} denotes the submatrix of all entries of $M$ at the intersection between a row of $R$ and a column of $C$.
A \emph{zone} of a matrix partitioned by $({\mathcal R},{\mathcal C})=(\{R_1, \ldots, R_k\},\{C_1, \ldots, C_\ell\})$ is any $R_i \cap C_j$ for $i \in [k]$ and $j \in [\ell]$.
A zone is \emph{constant} if all its entries are identical.
The \emph{error value} of a row-part $R_i$ (resp.~a column-part $C_j$) is the number of non-constant zones among all zones in $\{R_i \cap C_1, \ldots, R_i \cap C_\ell\}$ (resp.~$\{R_1 \cap C_j, \ldots, R_k \cap C_j\}$).
The \emph{error value} of $(\mathcal R,\mathcal C)$ is the maximum error value of a part, taken over all parts $R_i$ and $C_j$. 
Now the twin-width of a matrix~$M$ can be equivalently defined as the minimum~$d$ for which~$M$ admits a coarsening sequence in which all partitions have error value at most~$d$.

We will work with particular partitions, called \emph{divisions}, where every part consists of a set of consecutive rows, or a set of consecutive columns. 
If the matrix is $d$-twin-ordered, there is a coarsening sequence with error value at most~$d$, in which all the partitions are divisions.
We call \emph{division sequence} such a coarsening sequence.

\subsubsection{Grid minor theorem for twin-width}

A~\emph{$(t,t)$-division} is a division $(\mathcal R,\mathcal C)$ such that $|\mathcal R|=|\mathcal C|=t$.
A~\emph{$t$-grid minor} is a $(t,t)$-division whose $t^2$ zones contains a non-zero entry.
As for the \textsc{Permutation Pattern} breakthrough algorithm of Guillemot and Marx~\cite{Guillemot14}, a crucial engine of twin-width is the following celebrated theorem by Marcus and Tardos. 
\begin{theorem}[\cite{MarcusT04}]\label{thm:marcustardos}
For every integer $t$, there is some $c_t$ such that every $n\times m$ $0,1$-matrix $M$ with at least $c_t\max(n,m)$ entries 1 has a $t$-grid minor.
\end{theorem}
Informally, if a matrix has sufficiently many entries 1, then there is a large grid structure where each cell is ``complicated''.
The current best bound for $c_t$, due to Cibulka and Kyn\v{c}l~\cite{Cibulka16}, is $8/3(t+1)^22^{4t}$.

To leverage Marcus-Tardos theorem in the dense regime, too, we modify the definition of ``complicated'' from ``containing a 1'' to ``being mixed''.
A zone is \emph{horizontal} if all its columns are equal (restricted to the zone), and vertical if all the rows are equal.
Equivalently each row (resp.~column) within a horizontal zone (resp.~vertical zone) consists of a repeated same entry.
Note that a zone is constant (consists of a same entry repeated) if it is horizontal and vertical.
A zone is \emph{mixed} if it is not horizontal nor vertical.

We can now introduce the notions of $t$-mixed minors and $t$-mixed freeness.
A~\emph{$t$-mixed minor} of a matrix $M$ is a $(t,t)$-division of $M$ such that every zone is mixed.
A matrix is \emph{$t$-mixed free} if it does not admit a $t$-mixed minor.
We showed that having small twin-width and admitting no large mixed minors are equivalent in the following sense.

\begin{theorem}[\cite{twin-width1}]\label{thm:gridtheorem}
 Let $\alpha$ be the alphabet size for the matrix entries, and $c_t := 8/3(t+1)^22^{4t}$.
\begin{compactitem}
\item Every $t$-twin-ordered matrix is $2t+2$-mixed free.
\item Every $t$-mixed free matrix has twin-width at most $4c_t \alpha^{4c_t+2}=2^{2^{O(t)}}$.
\end{compactitem}
\end{theorem}
The first item is a relatively simple observation.
The difficulty lies in the second item.
In a nutshell, if the matrix is $t$-mixed free, we find, using Marcus-Tardos theorem, a sequence of divisions with small number of mixed zones per column and per row.
From this favorable sequence of divisions, we are able to extract an $f(t)$-contraction sequence.

One simple but important ingredient is a local characterization of mixedness by means of a \emph{corner}.
A~\emph{corner} in a matrix $M=(m_{i,j})_{i,j}$ is a mixed zone made by four contiguous entries $m_{i,j}, m_{i+1,j}, m_{i,j+1}, m_{i+1,j+1}$.
A~{0,1-corner} is a corner where each entry is in $\{0,1\}$.
\begin{lemma}[\cite{twin-width1}]\label{lem:corner}
A matrix is mixed if and only if it contains a corner.
\end{lemma}

In \cref{sec:small} we will work with specifically divided $0,1,r$-matrices, respecting the following invariants.
Every zone is filled with $r$ entries, or is \emph{non-mixed} (that is, horizontal or vertical) and has only $0$ and $1$ entries.
In this context, we will redefine the mixed zones as those filled with $r$ entries.
The coarsenings will be followed by updating the entries of the matrix to keep the invariants.
Namely every zone with a $0,1$-corner is filled with $r$ entries.
This new viewpoint mixes contraction sequence and coarsening sequence.
It will turn out useful to find, in a $t$-mixed free matrix, not just one ``good contraction'' (as in \cref{thm:gridtheorem}) but a linear number of disjoint pairs of ``good contractions''.
This will have two main consequences.
It will enable us to show that bounded twin-width classes are \emph{small} (see~\cref{sec:prelim:small} for a formal definition).
This will also be used to find $O(\log n)$-bits adjacency labeling schemes (see~\cref{sec:prelim:als}) for $n$-vertex graphs in classes of bounded twin-width. 

\subsubsection{Closure by FO transduction}

Bounded twin-width behaves surprisingly well with respect to first-order logic.
In addition to the fixed-parameter tractable algorithm running in time $f(d,|\phi|)n$ for model checking a first-order sentence $\phi$ on an $n$-vertex graph given with a $d$-contraction sequence, we show that bounded twin-width is preserved by first-order (FO) transductions.

\begin{theorem}[\cite{twin-width1}]\label{thm:transduction}
  Every transduction of a bounded twin-width class has bounded twin-width.
\end{theorem}

A formal definition of FO transductions can be found in several papers (see for instance~\cite{Blumensath10,twin-width1}).
As this definition is somewhat lengthy and technical and we will only use \cref{thm:transduction} in a black-box fashion, we refer the interested reader to these papers.
Informally an FO transduction of a graph $G$ defines several new graphs.
It consists of a non-deterministic ``coloring'' of $V(G)$ by a constant number of unary relations, followed by a redefinition of the edges by means of a fixed FO formula using the former edge predicate as well as these new unary relations.
The unary relations are then discarded, and we here further allow to take any induced subgraph of the obtained graph (to preserve the class heredity).
An FO transduction of a class $\mathcal C$ is simply the union of the graphs obtained by FO transduction of~$G$, for every $G \in \mathcal C$. 

\subsection{Small classes and the small conjecture}\label{sec:prelim:small}

We recall that a hereditary class is a class closed under taking induced subgraphs.
Formally if $G$ is in a hereditary class~$\mathcal C$, then for every induced subgraph $H$ of $G$, it also holds that $H$ is in $\mathcal C$.
The overwhelming majority of the usually considered classes of graphs are hereditary.\footnote{Notable exceptions include regular graphs, connected graphs, and visibility graphs of a point set.}

A class of graph $\mathcal C$ is said \emph{small} (resp.~\emph{factorial}), if there exists a constant $c$, such that the number of $n$-vertex graphs of $\mathcal C$ is at most $n!c^n$ (resp.~$n!^c=2^{O(n \log n)}$), for every $n \in \mathbb N$.
Recall that our $n$-vertex graphs are all assumed to be on the vertex set $[n]$, and that we count up to equality and \emph{not} up to isomorphism.
Norine et al.~\cite{Norine06} show that the number of $K_t$-minor free graphs on $[n]$ is at most $n!c^n$, for some integer $c$ depending only on~$t$.
In other words, proper minor-closed classes are small.
Marcus-Tardos theorem~\cite{MarcusT04}, combined with an argument due to Klazar~\cite{Klazar00}, implies that the number of $n \times n$ $0,1$-matrices avoiding a fixed permutation submatrix is at most $c^n$, for some constant $c$.
In particular the number of permutations on $n$ elements avoiding a fixed permutation grows in $2^{O(n)}$.
A translation of this result to graphs is that proper subclasses of permutation graphs are small.

We say that a class $\mathcal C$ has bounded twin-width if there exists an integer $d_{\mathcal C}$ such that every member of $\mathcal C$ has twin-width at most~$d_{\mathcal C}$.
Thus $\tww(\mathcal C) := \underset{G \in \mathcal C}{\sup}~\tww(G) < \infty$.

One of the main contributions of the paper is the following.
\begin{theorem}\label{thm:main}
  Every class with bounded twin-width is small.
\end{theorem}
This generalizes the smallness of proper minor-closed classes~\cite{Norine06}, proper subclasses of permutation graphs~\cite{MarcusT04,Klazar00}, and graphs with bounded treewidth \cite{Beineke69}, as we previously showed that all these classes have bounded twin-width~\cite{twin-width1}.
We then explore a possible converse for \cref{thm:main}.
Of course it is easy to artificially build an unbounded twin-width class with only $n!$ graphs of size $n$.
For example, by taking in the class a single (up to isomorphism) $n$-vertex graph among the $n$-vertex graphs with maximum twin-width, for every $n$. 
However this is not a satisfactory counterexample.
In combinatorics, classes of objects are often required to be closed under substructures.
For instance, a \emph{class of permutations} is by definition closed under taking subpermutations.
Same goes for graphs: Hereditary classes have richer properties than non-hereditary ones.
Many interesting questions on hereditary classes have trivial answers or are not even well-defined on general classes.

We provocatively conjecture the following converse of~\cref{thm:main}.
\begin{conjecture}[small conjecture]\label{conj:small}
  Every small hereditary class has bounded twin-width.
\end{conjecture}

It may seem ambitious to expect that the converse of~\cref{thm:main} holds for hereditary classes.
Why would the mere limited number of graphs guarantee anything close to a $d$-contraction sequence?
A typical example of a class with unbounded twin-width contains an infinite sequence of graphs $G_1, G_2, \ldots$ where every distinct pair $u, v \in V(G_i)$ satisfies $\card{N_{G_i}(u) \triangle N_{G_i}(v)} \geqslant i$.
Indeed any first contraction in $G_i$ creates a vertex with red degree at least $i$.
A class is said to have \emph{unbounded symmetric difference} if it contains such a sequence, and \emph{bounded symmetric difference}, otherwise.
So for every class $\mathcal C$ with bounded symmetric difference, there is an integer $d$ such that for every graph $G \in \mathcal C$, there exist two distinct vertices $u, v \in V(G)$ satisfying $\card{N(u) \triangle N(v)} \leqslant d$.
For example, the $i \times i$ rook graphs (with vertex set $[i] \times [i]$ and an edge between $(a,b)$ and $(c,d)$ if $a=c$ or $b=d$), with $i \geqslant 3$, is a class with unbounded symmetric difference.
However the hereditary closure of this class is not small.

Having bounded symmetric difference is a prerequisite to having bounded twin-width.
A~first step towards~\cref{conj:small} would be to show that small hereditary classes have bounded symmetric difference.
Even that is unclear.
For $K_{2,2}$-free classes or classes with girth at least 5, bounded symmetric difference simply implies bounded minimum degree.
Thus a very particular case of \cref{conj:small} is that there every small $K_{2,2}$-free hereditary class has bounded minimum degree.  

Let us present some elements supporting the conjecture.
First and foremost, bounded twin-width seems to ``stop at the right place'' in the sparse and dense realms.
Unit interval graphs (a small class) have bounded twin-width while interval graphs (a non-small class) do not.
Similarly among sparse classes, proper minor-closed classes (small) have bounded twin-width, whereas subcubic graphs (non-small) have unbounded twin-width.
We will also see that some expander classes have bounded twin-width (and are small), unlike random cubic graphs. 

An interesting test is the case of the $s$-subdivisions (where each edge of a graph is subdivided $s \geqslant 1$ times).
Since the number of subcubic graphs on $[n]$ is $n^{3n/2+O(n/\log n)}$, the $o(\log n)$-subdivisions of subcubic graphs is still a non-small class.
Thus by~\cref{thm:main}, they have unbounded twin-width.
We show a more fine-grained version of that fact by a direct proof.
We also build in polynomial time $O(1)$-sequences for $\Omega(\log n)$-subdivisions of $K_n$, which yields the following. 
\begin{theorem}\label{thm:subd}
  The $s$-subdivision of $K_n$ has bounded twin-width if and only if $s = \Omega(\log n)$.
  More precisely, for every integer $d$, there are $\ell_d < u_d$ such that the $\lfloor c \log n \rfloor$-subdivision of $K_n$ has twin-width at least $d$ for every $1 \leqslant c \leqslant \ell_d$, and at most $d$ for every $c \geqslant u_d$.    
\end{theorem}

The hereditary closure of $\Omega(\log n)$-subdivisions of $K_n$ is indeed a small class.
But \cref{thm:subd} in particular implies that this class does have bounded twin-width.
Dvořák and Norine~\cite{Dvorak10} show that, for any constants $c, \varepsilon > 0$, classes with expansion $r \mapsto c^{r^{1/3 - \varepsilon}}$ are small, while the class of all graphs with expansion $r \mapsto 6 \cdot 3^{\sqrt{r \log{(r+e)}}}$ is not small.
If the small conjecture is true, then bounded twin-width contains polynomial expansion (actually even expansion $r \mapsto 2^{r^{0.33}}$). 
Thus another possible first step to \cref{conj:small} is to show that bounded twin-width classes have polynomial expansion.

A supplementary motivation for the small conjecture appears if its proof is algorithmic, that is, yields on any small hereditary class a polytime algorithm which takes any graph of the class and outputs a (non-necessarily optimal) $O(1)$-sequence.
In light of~\cref{thm:main} and considering that $\omega(1)$-sequences are not as algorithmically useful, that would be almost as good as a constant approximation of twin-width in general graphs.

\subsection{Implicit representations}\label{sec:prelim:als}

A class $\mathcal C$ has an \emph{$f(n)$-bits adjacency labeling scheme} (or simply \emph{labeling scheme}, for short) if there is a decoding function $A: \{0,1\}^* \times \{0,1\}^* \rightarrow \{0,1\}$ such that for every $n$-vertex graph $G \in \mathcal C$ there is a labeling function $\ell: V(G) \rightarrow \{0,1\}^*$, satisfying $|\ell(u)| \leqslant f(n)$ for every $u \in V(G)$, and $A(\ell(u),\ell(v))=1$ if and only if $uv \in E(G)$.
Here we will further impose that the labeling function $\ell$ is injective.
For example trees now have $\log n + O(1)$-bits adjacency labeling scheme \cite{Alstrup17}, which up to the constant term, is optimal.
It is known that a class $\mathcal C$ has a $c \log n$-bits adjacency labeling scheme if and only if, for every integer $n$, there is a \emph{universal graph} graph $U_n$ (not necessarily in $\mathcal C$) on at most $n^c$ vertices such that every $n$-vertex graph of $\mathcal C$ is an induced subgraph of $U_n$~(see for instance \cite{Spinrad03}).
This becomes apparent when one considers the possible labels as the vertex set of the universal graph. 

Several classes, such as interval graphs and $K_t$-minor free graphs, are known to have $O(\log n)$-bits labeling schemes.
By a direct counting argument, only factorial classes can expect to admit $O(\log n)$-bits labeling scheme.
Indeed the number of distinct labels is $2^{O(\log n)}=n^{O(1)}$.
Thus the number of $n$-vertex graphs that can be induced subgraphs of the universal graph is only ${n^{O(1)} \choose n}=n^{O(n)}$.
The \emph{implicit graph conjecture} asserts that every factorial hereditary class has an $O(\log n)$-bits labeling scheme~\cite{Kannan92}.
We show the conjecture in the particular case of bounded twin-width classes. 

\begin{theorem}\label{thm:als}
  Every bounded twin-width class admits an $O(\log n)$-bits labeling scheme.
\end{theorem}

This result is at the same time quite strong and quite weak.
Its strength lies in its broad generality.
We produce a unified labeling scheme for very different sparse and dense classes. 
However there are two caveats, both linked to its generality.
The first one is that we still do not know if the labeling function can be computed in polynomial time.
Indeed it requires a $d$-sequence (even a so-called \emph{parallel $D$-sequence} of logarithmic length).
If we know how to compute this sequence in many bounded twin-width classes, we do not know in the full generality of all the graphs with twin-width at most~$d$.
In the latter case, we currently need exponential time to find the sequence, and then to compute the labeling.
The decoding function, that is the adjacency test, runs in time $O(\log n)$ in the RAM model with unit-cost arithmetic operations over words of logarithmic length.
The second caveat is that when restricted to particular classes, the multiplicative constant preceding $\log n$ given by our proof is much larger than in the shortest known labeling schemes.
For instance, the current best labeling scheme for $K_t$-minor free graphs requires $2 \log n + o(\log n)$ bits per vertex \cite{Gavoille07}, while our multiplicative constant is double-exponential in $t$.

Improving the constant $c$ of existing $(c+o(1)) \log n$-bits labeling schemes is topical in implicit representations.  
Recently planar graphs were shown to admit a $(1+o(1)) \log n$-bits adjacency labeling scheme~\cite{Dujmovic-als}.
It is optimal up to the second-order term.
The labeling scheme is actually more general, and works for all subgraphs of strong products $H \boxtimes P$ where $H$ is a bounded-treewidth chordal graph (or \emph{$k$-tree}, for some fixed~$k$), and $P$ is a path.
A class $\mathcal C$ is said~\emph{flat} if there is an integer $k$ such that $\mathcal C \subseteq \sub(\mathcal H \boxtimes \mathcal P)$ where $\mathcal P$ is the set of all paths, and $\mathcal H$ is a set of graphs with treewidth at most~$k$.
An ongoing program (not specific to adjacency labeling schemes), dubbed \emph{graph product structure theorem}, established that many small and sparse classes are flat.  
This was initiated by a paper by Pilipczuk and Siebertz~\cite{PilipczukS19} showing a similar result for planar graphs.
This property was extended to apex-minor free~\cite{Dujmovic-queue}, bounded-degree minor-free~\cite{Dujmovic-clustered}, and $k$-planar classes~\cite{Dujmovic-kplanar}.
Hence they all enjoy a $(1+o(1)) \log n$-bits adjacency labeling scheme.
Interestingly all these classes have bounded twin-width (minor-free classes and $k$-planar graphs have bounded twin-width~\cite{twin-width1}).
This is no coincidence.
We will see that the strong product of two bounded twin-width graphs, one of which has bounded degree, has bounded twin-width. 
\begin{restatable}{theorem}{product}\label{thm:product-stability}
  Let $G$ and $H$ be two graphs.
  Then $\tww(G \boxtimes H) \leqslant \max\{\tww(G)(\Delta(H)+1)+2\Delta(H),\tww(H)+\Delta(H)\}$.
\end{restatable}

As cliques have twin-width 0, taking subgraphs does in general not preserve twin-width at all.
Nevertheless on ``sparse'' classes, bounded twin-width is subgraph-closed.
We show that if the strong product of a bounded twin-width class $\mathcal G$ with a bounded-degree bounded twin-width class $\mathcal H$ is $K_{t,t}$-free, then the subgraphs of $\mathcal G \boxtimes \mathcal H$ have bounded twin-width.
\begin{restatable}{theorem}{subproduct}\label{thm:product-stability-class}
Let $\mathcal G$ and $\mathcal H$ two classes such that $\mathcal G \boxtimes \mathcal H$ is $K_{t,t}$-free.   
Then $\tww(\sub(\mathcal G \boxtimes \mathcal H)) \leqslant f(\tww(\mathcal G),\tww(\mathcal H),\Delta(\mathcal H),t)$.
\end{restatable}
In particular flat classes have bounded twin-width (since graphs with bounded treewidth have bounded twin-width, and flat classes are $K_{t,t}$-free).
By essence, the ``flat class'' approach to $(1+o(1)) \log n$-bits labeling scheme is limited to classes that are $K_t$-free.
Another interesting limit case is minor-free classes which are not apex-minor free, like all the $K_6$-minor free graphs for example.
Dujmovi\'c et al.~\cite{Dujmovic-queue} show that these classes are not flat.

We hope that the \emph{versatile tree of contractions} (see~\cref{lem:versatile-tww}) or the \emph{short parallel contraction sequence} (see~\cref{lem:short-d-sequence}) may help for small dense classes and $K_t$-minor free graphs.
We optimistically conjecture that our~\cref{thm:als} can be improved to an optimal labeling scheme up to the second-order term.
\begin{conjecture}\label{conj:als}
Every bounded twin-width class has a $(1+o(1)) \log n$-bits labeling scheme.  
\end{conjecture}

\subsection{Sparse twin-width}

The trace of bounded twin-width on sparse classes is also an interesting and potentially new class.
There are five natural ways of forcing a bounded twin-width class to be ``sparse'': forbidding $K_{t,t}$ as a subgraph, forbidding a $d$-grid minor in its adjacency matrix (and not a mere $d$-mixed minor), requiring that every graph has bounded average degree, requiring that the subgraphs also have bounded twin-width, and requiring that the class has bounded expansion. 
Let $A_\sigma(G)$ denote the adjacency matrix of $G$ when $V(G)$ is ordered by $\sigma$.
We say that a class $\mathcal C$ is \emph{$d$-grid free} if for every $G \in \mathcal C$ there is an ordering $\sigma$ of $V(G)$ such that $A_\sigma(G)$ is $d$-grid free.

We show that all five definitions are actually equivalent.
\begin{restatable}{theorem}{sparse}\label{thm:sparseboundedtww}
  If $\mathcal C$ is a hereditary class of bounded twin-width, the following are equivalent.
  \begin{compactitem}
  \item(i) There is an integer $t$ such that no graph of $\mathcal C$ contains $K_{t,t}$ as a subgraph.
  \item(ii) There is an integer $d$ such that $\mathcal C$ is $d$-grid free.
  \item(iii) There is an integer $g$ such that every $n$-vertex graph $G \in \mathcal C$ has at most $gn$ edges.
  \item(iv) The subgraph closure $\sub(\mathcal C)$ has bounded twin-width.
  \item(v) There is a function $f$ such that $\nabla_r(\mathcal C) \leqslant f(r)$ for every $r$.
  \end{compactitem}
\end{restatable}

Ignoring item $(iv)$, a compact version of this theorem reads: For a hereditary class of bounded twin-width having bounded grid minors, bicliques, average degree, or expansion are all equivalent. 

Thus we say that a hereditary class has \emph{bounded sparse twin-width} if it has bounded twin-width and satisfies any of the five items (that is, satisfies all five).
One may wonder whether bounded sparse twin-width coincides with some existing sparse class.
More generally it is interesting to see how bounded sparse twin-width compares to the established sparse classes.
A few candidates come to mind: polynomial expansion, bounded expansion, bounded queue number, bounded stack number, bounded nonrepetitive coloring classes.
Although we do not prove it for bounded queue or stack number, we argue that these classes do not coincide with bounded sparse twin-width.

As cubic graphs have unbounded twin-width, bounded expansion is strictly more general than bounded sparse twin-width.
For the same reason, bounded nonrepetitive coloring does not imply bounded sparse twin-width.
It is possible however that bounded sparse twin-width classes have bounded nonrepetitive coloring.
The existence of an infinite family of cubic expanders with bounded twin-width implies that bounded sparse twin-width classes do not necessarily have polynomial expansion.
If the small conjecture is true, polynomial expansion would be a strict subset of bounded sparse twin-width. 
We will show that classes with bounded queue number or bounded stack number have bounded (sparse) twin-width.
We believe that this inclusion is strict and that the expanders based on random 2-lifts have unbounded queue and stack numbers.

\subsection{Organization of the rest of paper}

In~\cref{sec:small} we show \cref{thm:main}, that every class of bounded twin-width is small.
From this we conclude that non-small classes such as subcubic graphs, interval graphs, and triangle-free unit segment graphs have unbounded twin-width.
This can be respectively put in perspective with the fact that some cubic expanders (as we see in \cref{sec:expanders}), unit interval graphs, and $K_t$-free unit $d$-dimensional ball graphs, have bounded twin-width~\cite{twin-width1}.
In~\cref{sec:adjacency} we leverage the results from the previous section to present $O(\log n)$-bits adjacency labeling schemes on bounded twin-width classes.
We then explore the converse of \cref{thm:main} for hereditary classes.
In~\cref{sec:expanders} we show that the small class of cubic expanders obtained by iterated 2-lifts from $K_4$ has indeed bounded twin-width.
In~\cref{sec:subd-cliques} we prove that the $s$-subdivision of the clique $K_n$, with $s>0$, has bounded twin-width if and only if $s = \Omega(\log n)$.
In~\cref{sec:sparse-tww} we prove \cref{thm:sparseboundedtww}, the list of characterizations of bounded sparse twin-width.
We then show that flat classes, and classes with bounded queue or stack number have bounded (sparse) twin-width.
In~\cref{sec:groups} we investigate the twin-width of the finite induced subgraphs of a fixed Cayley graph.
We show that such classes are small for every finitely generated group.
This is a rare example of a small class for which we still do not know if the twin-width is bounded.

\section{Bounded twin-width classes are small}\label{sec:small}

In this section we show that graphs of bounded twin-width have bounded \emph{versatile twin-width}.
Informally it says that whenever we can find a sequence (or path) of $d$-contractions, we can even find a tree of $D$-contractions with linear arity, for some $D$ bounded by a function of~$d$. 
This result is fairly technical but shares some ideas and arguments with Section 5 of our previous paper~\cite{twin-width1}.
We made the current section self-contained.
We nevertheless mention some frequent parallels with~\cite{twin-width1}.
Finally we can follow the end of the proof of Norine et al.~\cite{Norine06} --that proper minor-closed classes are small-- to extend the result to bounded twin-width.

\subsection{The proof for proper minor-closed classes and how (not) to tune it}\label{subsec:minor-closed}
Let us first give a brief sketch of Norine et al.'s proof, which works by induction on~$n$.
They say that a vertex is \emph{$d$-good} if it has degree at most~$d$ and either has a twin or has a neighbor with degree at most~$d$.
They show the following technical lemma: $K_t$-minor free $n$-vertex graphs have at least $n/d$ $d$-good vertices, for some $d$ function of $t$ only.
Let $\mathcal I_{n,t}$ be the set of $K_t$-minor free graphs on $[n]$, and $\mathcal K_{n,t}$ be the subset of all those graphs of $\mathcal I_{n,t}$ where vertex $n$ is $d$-good.
By their lemma $n/d \cdot |\mathcal I_{n,t}| \leqslant n |\mathcal K_{n,t}|$, hence $|\mathcal I_{n,t}| \leqslant d |\mathcal K_{n,t}|$.
Furthermore, any graph of $\mathcal K_{n,t}$ admits an index $i \in [n-1]$ such that either $i$ and $n$ are false twins, or $i$ and $n$ are adjacent and have at most $d-1$ other neighbors each.
Therefore any $G \in \mathcal K_{n,t}$ can be obtained from a $G' \in \mathcal I_{n-1,t}$ and $i \in [n-1]$ by either introducing a new vertex $n$ false twin of $i$ (one graph), or by splitting $i$ into $i$ and a new vertex $n$ adjacent to $i$, and by distributing in $G$ the at most $2(d-1)$ neighbors of $i$ in $G'$ into: neighbors of $i$ only, neighbors of $n$ only, and common neighbors (at most $3^{2(d-1)}$ graphs).
Hence $|\mathcal I_{n,t}| \leqslant d(1+3^{2(d-1)})(n-1)|\mathcal I_{n-1,t}| \leqslant (1+3^{2(d-1)})(n-1) \cdot (n-1)!c^{n-1} \leqslant n!c^n$, by taking $c := d(1+3^{2(d-1)})$.

We need to redefine the notion of being $d$-good for bounded twin-width classes.
A very natural candidate for that would be to say that a vertex is $d$-good if it admits a $d$-contraction with another vertex.
After all, there is always such a vertex (or such a pair of vertices) in a $d$-trigraph.
However, we cannot expect $d$-trigraphs to have linearly many such vertices.
Think for instance of a path on $n$ vertices.
It has twin-width 1, but only four vertices (the two endpoints and their neighbor) that can be contracted to yield a 1-sequence.
Surely we could allow mere $D$-contractions, for some $D \gg d$, but then we would leave the class of $d$-trigraphs.
So it would be unclear which class we are bounding the size of.
It is indeed noteworthy in the above sketch that by deleting a vertex or contracting adjacent vertices, one remains in the class of $K_t$-minor free graphs.

To overcome that issue, we introduce a more robust notion of bounded twin-width.
A \emph{tree of $d$-contractions} of a $d$-trigraph $G$ is a rooted tree, whose root is labeled by $G$, and whose leaves are all labeled by 1-vertex graphs $K_1$, and such that one can go from any parent to any child by a $d$-contraction.
With this new definition, $d$-sequences coincide with trees of $d$-contractions which are in fact paths.
We say that a trigraph $G$ has \emph{versatile twin-width~$d$} if there exists some~$p$, function of $d$ only, such that $G$ admits a tree of $d$-contractions in which every internal node has at least $|V(\cdot)|/p$ children with distinct labels (where $|V(\cdot)|$ denotes the number of vertices of the corresponding node label).
Such a tree is then called a \emph{versatile tree of $d$-contractions}.

Let us say that a contraction is \emph{$d$-correct} (or simply \emph{correct} when we precise that it is a $d$-contraction) if the obtained graph has twin-width at most $d$.
The inductive nature of versatile twin-width provides us the desired stability.
Not only $G$ admits linearly many correct $d$-contractions, but it admits linearly many $d$-contractions towards graphs of versatile twin-width $d$.
This is indeed witnessed by the subtrees rooted at each child of the root labeled by $G$.
We now focus on proving that every trigraph with twin-width $d$ has a versatile tree of $D$-contractions, for a larger $D$ function of $d$ only.
This is a bit technical, but once it is done, we will be able to mimic the end of Norine et al.'s proof.

\subsection{Neatly divided symmetric $0,1,r$-matrices}

Recall that $r$ (for red) is the error symbol.
It will now be convenient to tune some of the notions developed in our previous paper specifically for $0,1,r$-matrices with particular divisions.
The notions introduced without a definition are all formalized in \cref{sec:prelim} of the present paper, as well as in~\cite[Section 5]{twin-width1}.
Reading first \cite[Section 5]{twin-width1} does not harm, but it is not necessary to understand the current section.

We will manipulate divisions of $0,1,r$-matrices such that every zone either contains only $r$~entries or contains no $r$ entry and is horizontal or vertical (or both).
Let us call \emph{neat} such a division.
Zones filled with $r$ entries are now called \emph{mixed}. 
A \emph{neatly divided matrix} is a pair $(M,(\mathcal R,\mathcal C))$ where $M$ is a $0,1,r$-matrix and $(\mathcal R,\mathcal C)$ is a neat division of $M$.
A \emph{$t$-mixed minor} in a neatly divided matrix is a $(t,t)$-division which coarsens the neat subdivision, and contains in each of its $t^2$ zones at least one mixed zone (filled with $r$ entries) or a 0,1-corner.
See~\cref{fig:new-mixed-minor} for an illustration.
A neatly divided matrix is said \emph{$t$-mixed free} if it does not admit a $t$-mixed minor. 

A \emph{mixed cut of a row-part $R \in \mathcal R$} of a neat division $(\mathcal R, \mathcal C=\{C_1,C_2,\ldots\})$ is an index $i$ such that both $R \cap C_i$ and $R \cap C_{i+1}$ are \emph{non}-mixed, and there is a $0,1$-corner in the 2-by-$|R|$ zone defined by the last column of $C_i$, the first column of $C_{i+1}$, and $R$.
Importantly, a mixed cut cannot border a mixed zone.
(This is a difference with the definition of~\cite[Section 5]{twin-width1}.)
The \emph{mixed value of a row-part $R \in \mathcal R$} of a neat division $(\mathcal R, \mathcal C=\{C_1,C_2,\ldots\})$ is the number of mixed zones $R \cap C_j$ plus the number of mixed cuts between two (adjacent non-mixed) zones $R \cap C_j$ and $R \cap C_{j+1}$.
Note that a mixed cut counts for one unit in the mixed value, regardless of the number of corners overlapping the two adjacent zones.
We similarly define the mixed value of a column-part $C \in \mathcal C$.
The \emph{mixed value of a neat division} of a $0,1,r$-matrix is the maximum of the mixed values taken over every part.
The \emph{part size} of a division (resp.~partition) $(\mathcal R,\mathcal C)$ is defined as $\max(\max_{R \in \mathcal R}|R|,\max_{C \in \mathcal C}|C|)$.
A division is \emph{symmetric} if the largest row index of each row-part and the largest column index of each column-part define the same set of integers, that is informally, if the horizontal separations are symmetric of the vertical separations about the main diagonal.
For instance the division depicted on \cref{fig:new-mixed-minor} is symmetric since both the largest row indices of the row-parts and the largest column indices of the column-parts define the set $\{2,3,4,6\}$. 
We call \emph{symmetric fusion} of a symmetric division the fusion of two consecutive parts in $\mathcal C$ and of the two corresponding parts in $\mathcal R$.
A symmetric fusion on a symmetric division yields another symmetric division.
A matrix $A := (a_{i,j})_{i,j}$ is said \emph{symmetric} in the usual sense, namely, for every entry $a_{i,j}$ of~$A$, $a_{i,j}=a_{j,i}$.

\begin{figure}[h!]
  \centering
  \begin{tikzpicture}
    \def\s{0.5}
    \def\hb{\s/2}
    \def\vb{\s/2}
    \def\he{8.5 * \s}
    \def\ve{8.5 * \s}
     \foreach \i/\j/\v in {1/1/1,1/2/1,1/3/r,1/4/r,1/5/0,1/6/0,1/7/1,1/8/1, 2/1/0,2/2/0,2/3/r,2/4/r,2/5/0,2/6/1,2/7/1,2/8/1, 3/1/1,3/2/1,3/3/0,3/4/0,3/5/0,3/6/1,3/7/1,3/8/0, 4/1/1,4/2/1,4/3/1,4/4/0,4/5/0,4/6/0,4/7/1,4/8/1, 5/1/r,5/2/r,5/3/1,5/4/1,5/5/0,5/6/0,5/7/1,5/8/0, 6/1/r,6/2/r,6/3/0,6/4/0,6/5/0,6/6/1,6/7/1,6/8/0, 7/1/0,7/2/0,7/3/1,7/4/1,7/5/1,7/6/0,7/7/r,7/8/r, 8/1/1,8/2/1,8/3/0,8/4/0,8/5/1,8/6/1,8/7/r,8/8/r}{
      \node (e\i\j) at (\s * \i,\s * \j) {$\v$} ;
    }
    \draw (\hb-0.05,\vb) -- (\hb-0.05,\ve) --++(0.1,0) ;
    \draw (\hb-0.05,\vb) --++(0.1,0) ;
    \draw (\he+0.05,\vb) -- (\he+0.05,\ve) --++(-0.1,0) ;
    \draw (\he+0.05,\vb) --++(-0.1,0) ;
    \foreach \i in {2,3,4,6}{
      \draw (\i * \s + \s/2,\vb) -- (\i * \s + \s/2,\ve) ;
    }
    \foreach \j in {2,4,5,6}{
      \draw (\hb,\j * \s + \s/2) -- (\he,\j * \s + \s/2) ;
    }
    \begin{scope}[xshift=7cm]
    \foreach \i/\j/\v in {1/1/1,1/2/1,1/3/r,1/4/r,1/5/0,1/6/0,1/7/1,1/8/1, 2/1/0,2/2/0,2/3/r,2/4/r,2/5/0,2/6/1,2/7/1,2/8/1, 3/1/1,3/2/1,3/3/0,3/4/0,3/5/0,3/6/1,3/7/1,3/8/0, 4/1/1,4/2/1,4/3/1,4/4/0,4/5/0,4/6/0,4/7/1,4/8/1, 5/1/r,5/2/r,5/3/1,5/4/1,5/5/0,5/6/0,5/7/1,5/8/0, 6/1/r,6/2/r,6/3/0,6/4/0,6/5/0,6/6/1,6/7/1,6/8/0, 7/1/0,7/2/0,7/3/1,7/4/1,7/5/1,7/6/0,7/7/r,7/8/r, 8/1/1,8/2/1,8/3/0,8/4/0,8/5/1,8/6/1,8/7/r,8/8/r}{
      \node (e\i\j) at (\s * \i,\s * \j) {$\v$} ;
    }
    \draw (\hb-0.05,\vb) -- (\hb-0.05,\ve) --++(0.1,0) ;
    \draw (\hb-0.05,\vb) --++(0.1,0) ;
    \draw (\he+0.05,\vb) -- (\he+0.05,\ve) --++(-0.1,0) ;
    \draw (\he+0.05,\vb) --++(-0.1,0) ;
    \foreach \i in {3,6}{
      \draw[very thick] (\i * \s + \s/2,\vb) -- (\i * \s + \s/2,\ve) ;
    }
    \foreach \i in {2,4}{
      \draw[thin] (\i * \s + \s/2,\vb) -- (\i * \s + \s/2,\ve) ;
    }
    \foreach \j in {4,6}{
      \draw[very thick] (\hb,\j * \s + \s/2) -- (\he,\j * \s + \s/2) ;
    }
    \foreach \j in {2,5}{
      \draw[thin] (\hb,\j * \s + \s/2) -- (\he,\j * \s + \s/2) ;
    }
    \foreach \i/\j in {1/3,1/5,2/7,5/1,5/5,4/7,7/2,7/5,7/7}{
    \draw[fill,fill opacity=0.05,red,dashed] (\i * \s - \s / 2+0.05, \j * \s - \s / 2+0.05) --++(0.85,0) --++(0,0.9) --++(-0.85,0) -- cycle ;}
    \end{scope}
  \end{tikzpicture}
  \caption{To the left, a neat division: each zone is horizontal, or vertical, or full with $r$~entries (mixed zone). 
    Note that the division is symmetric but not the matrix.
    To the right, in bold, a 3-mixed minor of the neat division.
    Observe that it coarsens the neat division and contains in each of its 9 zones either a 0,1-corner or a mixed zone (framed by red dashed boxes).}
  \label{fig:new-mixed-minor}
\end{figure}
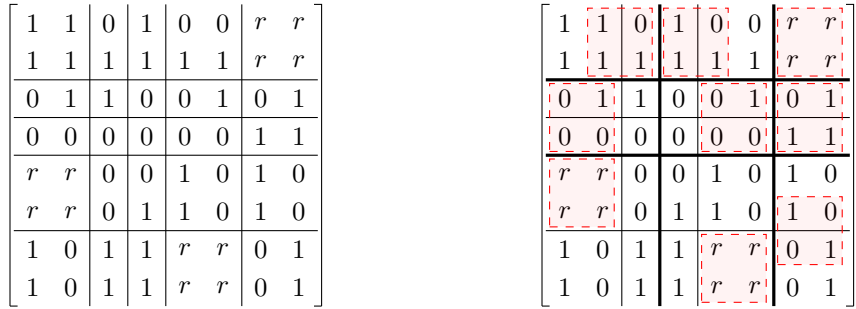

The following definition is crucial.
It lists the invariants that we want to keep in our neatly divided matrices in order to build a versatile tree of contractions.

\begin{definition}
  Let $\mathcal M_{n,d}$ be the class of the neatly divided $n \times n$ symmetric $0,1,r$-matrices $(M,(\mathcal R,\mathcal C))$, such that $(\mathcal R,\mathcal C)$ is symmetric and has:
  \begin{itemize}
  \item mixed value at most $4c_d$,
  \item part size at most $2^{4c_d+2}$, and
  \item no $d$-mixed minor.
  \end{itemize}
\end{definition}

In the previous definition, $c_d := 8/3(d+1)^22^{4d}$ as defined in the improvement of Marcus-Tardos bound~\cite{Cibulka16}.
The conditions of the first and second bullets are enough to bound the red number of a neatly divided matrix of $\mathcal M_{n,d}$.
\begin{lemma}\label{lem:ndm-red-number}
  Let $(M,(\mathcal R,\mathcal C))$ be in $\mathcal M_{n,d}$.
  The red number of $M$ is at most $4c_d \cdot 2^{4c_d+2}$.
\end{lemma}
\begin{proof}
  Any row or column intersects at most $4c_d$ mixed zones (filled with $r$ entries).
  Each mixed zone has width and length bounded by the part size $2^{4c_d+2}$.
  Hence the maximum total number of $r$ entries on a single row or column is at most $4c_d \cdot 2^{4c_d+2}$.
\end{proof}

\subsection{Finding invariant-preserving coarsenings}

A \emph{coarsening} of a neatly divided matrix $(M,(\mathcal R,\mathcal C))$ is a neatly divided matrix $(M',(\mathcal R',\mathcal C'))$ such that $(\mathcal R',\mathcal C')$ is a coarsening of $(\mathcal R,\mathcal C)$, and $M'$ is obtained from $M$ by setting to $r$ all entries that lie, in $M$ divided by $(\mathcal R',\mathcal C')$, in a zone with at least one $r$ entry or a 0,1-corner.
We also refer to the process of going from $(M,(\mathcal R,\mathcal C))$ to $(M',(\mathcal R',\mathcal C'))$ as \emph{coarsening operation} (or simply \emph{coarsening}). 
A coarsening operation from $(M,(\mathcal R,\mathcal C)) \in \mathcal M_{n,d}$ to $(M',(\mathcal R',\mathcal C'))$ is said \emph{invariant-preserving} if $(M',(\mathcal R',\mathcal C')) \in \mathcal M_{n,d}$, and \emph{elementary} if it consists of a single symmetric fusion.
The following lemma shows that not having a $t$-mixed minor is preserved for free in coarsenings of neatly divided matrices.

\begin{lemma}\label{lem:coarsening-mixed-free}
  Every coarsening $(M',(\mathcal R',\mathcal C'))$ of a $t$-mixed neatly divided free matrix $(M,(\mathcal R,\mathcal C))$ is $t$-mixed free itself.
\end{lemma}
\begin{proof}
  Assume there is a $t$-mixed minor $(\mathcal R^*,\mathcal C^*)$ of $(M',(\mathcal R',\mathcal C'))$.
  Let us consider the $(t,t)$-division $(\mathcal R^*,\mathcal C^*)$ in $(M,(\mathcal R,\mathcal C))$.
  By transitivity, $(\mathcal R^*,\mathcal C^*)$ coarsens $(\mathcal R,\mathcal C)$.

  There are two possibilities for a zone $Z$ of $(\mathcal R^*,\mathcal C^*)$ in $(M',(\mathcal R',\mathcal C'))$.
  Either it contains a $0,1$-corner, but then, $Z$ contains the same $0,1$-corner in $(M,(\mathcal R,\mathcal C))$.
  This is because the coarsening operation of a neatly divided matrix never replaces entries by 0 or 1 entries (we may only add $r$ entries).
  Or $Z$ contains an $r$ entry, or more precisely a zone $Z' \subseteq Z$ of $(M',(\mathcal R',\mathcal C'))$ filled with $r$ entries. 
  Either one of these $r$ entries was already present in $(M,(\mathcal R,\mathcal C))$, or the $r$ entries of $Z'$ appear after the fusion of a zone $Z_1 \subseteq Z'$ adjacent to a zone $Z_2 \subseteq Z'$ such that $Z_1 \cup Z_2$ contained a $0,1$-corner (and $Z_1 \cup Z_2 \subseteq Z' \subseteq Z$).
  
  Therefore, in any case, $Z$ in $(M,(\mathcal R,\mathcal C))$ contains an $0,1$-corner or an $r$ entry.
  We conclude that $(\mathcal R^*,\mathcal C^*)$ is a $t$-mixed minor of $(M,(\mathcal R,\mathcal C))$.
\end{proof}

The previous lemma will in particular give us some control on the average mixed value among the parts of a coarsening of a neatly divided matrix in $\mathcal M_{n,d}$.
This turns out crucial to find a coarsening which preserves the imposed upper bound on the overall mixed value.

\begin{lemma}\label{lem:coarsening-amv}
  Let $(M,(\mathcal R,\mathcal C))$ be in $\mathcal M_{n,d}$, and $(M',(\mathcal R,\mathcal C'))$ be a coarsening of $(M,(\mathcal R,\mathcal C))$ with $|\mathcal C'| \geqslant \lceil |\mathcal C|/2 \rceil$.
  Then the average mixed value among all the parts of $\mathcal C'$ on $(M',(\mathcal R,\mathcal C'))$ is at most $2c_d$.
\end{lemma}
\begin{proof}
  Note that in the coarsening operation of the lemma statement, we made only fusions in $\mathcal C$.
  Naturally the same would symmetrically work if only fusions in $\mathcal R$ were made.

  Let us assume by contradiction that the average mixed value $\gamma$, taken among every part $C \in \mathcal C'$ on $(M',(\mathcal R,\mathcal C'))$ is strictly greater than $2c_d$.
  We consider two coarsenings of $(M',(\mathcal R=\{R_1,\ldots,R_h\},\mathcal C'))$: $(M'_1,(\mathcal R_1=\{R_1 \cup R_2,R_3 \cup R_4,\ldots\},\mathcal C'))$ and $(M'_2,(\mathcal R_2=\{R_1, R_2 \cup R_3,R_4 \cup R_5,\ldots\},\mathcal C'))$.
  Let $C \in \mathcal C'$ be any part with $a$ mixed zones and $b$ mixed cuts on $(M',(\mathcal R,\mathcal C'))$.
  Let $a_1$, respectively $a_2$, be the number of mixed zones of $C$ on $(M',(\mathcal R_1,\mathcal C'))$, respectively on $(M',(\mathcal R_2,\mathcal C'))$.
  We claim that $a+b \leqslant a_1+a_2$.

  Indeed we can design the following injection from the mixed zones and cuts of $C$ on $(M',(\mathcal R,\mathcal C'))$ to the mixed zones of $C$ on $(M',(\mathcal R_1,\mathcal C'))$ and on $(M',(\mathcal R_2,\mathcal C'))$.
  We order the mixed zones and cuts of $C$ on $(M',(\mathcal R,\mathcal C'))$ from top to bottom, say, $x_1, x_2, \ldots, x_{a+b}$.
  For $i$ going from 1 to $a+b$, we attribute $x_i$ to $(\mathcal R_1,\mathcal C')$ or to $(\mathcal R_2,\mathcal C')$ based on the following rules.
  If $x_i$ is a mixed cut, there is a unique $j \in \{1,2\}$ such that $x_i$ is contained in a mixed zone of $C$ on $(M',(\mathcal R_j,\mathcal C'))$, so we map $x_i$ to this mixed zone.
  If $x_i$ is a mixed zone, we map it to the mixed zone containing $x_i$ in $(M',(\mathcal R_{3-j},\mathcal C'))$, where $x_{i-1}$ was mapped to a mixed zone in $(M',(\mathcal R_j,\mathcal C'))$.
  This is possible since there is a zone containing $x_i$ in both $(M',(\mathcal R_1,\mathcal C'))$ and $(M',(\mathcal R_2,\mathcal C'))$.
  For $x_1$ to be well-defined, we can imagine that there is a fictitious $x_0$ attributed to $(\mathcal R_2,\mathcal C')$.
  To see that this is indeed an injection we first need to recall that there is no mixed cut bordering a mixed zone.
  Suppose on the contrary that a same mixed zone $Z$ of $C$ on, say, $(M',(\mathcal R_1,\mathcal C'))$ has two preimages $x_i$ and $x_{i'}$, with $i<i'$.
  If $x_i$ and $x_{i'}$ are mixed zones, they need to be consecutive to both be in $Z$, hence $i'=i+1$.
  But then $x_{i'}$ should have been attributed to $(M',(\mathcal R_2,\mathcal C'))$ according to our rules.
  As $Z$ contains at most one mixed cut of $C$ on $(M',(\mathcal R,\mathcal C'))$, $x_i$ and $x_{i'}$ cannot be both mixed cuts.
  Finally it is impossible that exactly one of $x_i, x_{i'}$ is a mixed zone (and the other a mixed cut), since it would imply a mixed cut incident to a mixed zone. 
  See \cref{fig:injection} for an illustration.
  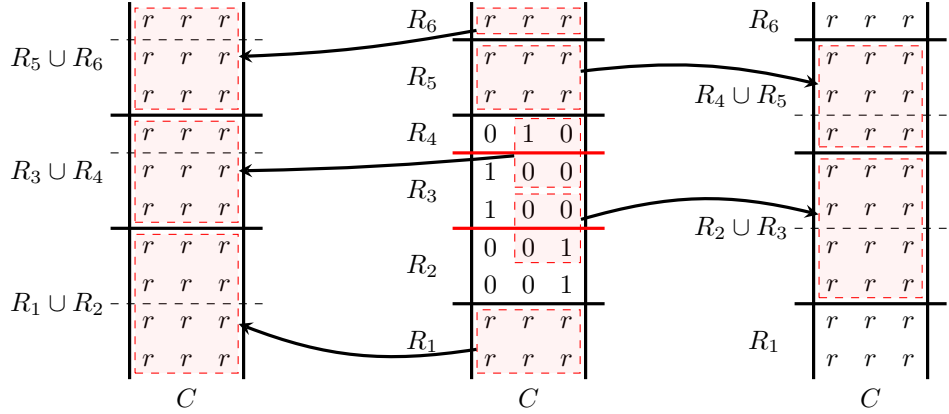
\begin{figure}
  \centering
  \begin{tikzpicture}
    \def\s{0.5}
    \def\hb{2 * \s}
    \def\vb{\s/2}
    \def\he{6 * \s}
    \def\ve{10.5 * \s}
    \foreach \i/\j/\v in {3/1/r,3/2/r,3/3/0,3/4/0,3/5/1,3/6/1,3/7/0,3/8/r,3/9/r,3/10/r, 4/1/r,4/2/r,4/3/0,4/4/0,4/5/0,4/6/0,4/7/1,4/8/r,4/9/r,4/10/r, 5/1/r,5/2/r,5/3/1,5/4/1,5/5/0,5/6/0,5/7/0,5/8/r,5/9/r,5/10/r}{
      \node (e\i\j) at (\s * \i,\s * \j) {$\v$} ;
    }
    \foreach \i in {2,5}{
      \draw[very thick] (\i * \s + \s/2,\vb) -- (\i * \s + \s/2,\ve) ;
    }
    \foreach \j in {2,7,9}{
      \draw[very thick] (\hb,\j * \s + \s/2) -- (\he,\j * \s + \s/2) ;
    }
      \foreach \j in {4,6}{
      \draw[red,very thick] (\hb,\j * \s + \s/2) -- (\he,\j * \s + \s/2) ;
    }
    \foreach \i/\j in {1.5/1,3.5/2,5.5/3,7/4,8.5/5,10/6}{
      \node at (0.6,\i * \s) {$R_\j$} ;
    }
    \node at (4 * \s,0) {$C$} ;

    \foreach \i/\j/\l in {e31/e52/x1,e44/e55/x2,e46/e57/x3,e38/e59/x4,e310/e510/x5}{
      \node[draw,fill,fill opacity=0.05,red,dashed,inner sep=-0.03cm,fit=(\i) (\j)] (\l) {} ;
    }
    
    \begin{scope}[xshift=-4.5cm]
      \foreach \i in {3,4,5}{
        \foreach \j in {1,...,10}{
          \node (e\i\j) at (\s * \i,\s * \j) {$r$} ;
          }
      }
    \foreach \i in {2,5}{
      \draw[very thick] (\i * \s + \s/2,\vb) -- (\i * \s + \s/2,\ve) ;
    }
    \foreach \j in {4,7}{
      \draw[very thick] (\hb,\j * \s + \s/2) -- (\he,\j * \s + \s/2) ;
    }
   \foreach \j in {2,6,9}{
      \draw[dashed] (\hb,\j * \s + \s/2) -- (\he,\j * \s + \s/2) ;
    }
    \foreach \i/\j/\k in {2.5/1/2,6/3/4,9/5/6}{
      \node at (0.3,\i * \s) {$R_\j \cup R_\k$} ;
    }
    \node at (4 * \s,0) {$C$} ;
     \foreach \i/\j/\l in {e31/e54/y1,e35/e57/y2,e38/e510/y3}{
      \node[draw,fill,fill opacity=0.05,red,dashed,inner sep=-0.03cm,fit=(\i) (\j)] (\l) {} ;
    }
    \end{scope}
    
     \begin{scope}[xshift=4.5cm]
       \foreach \i in {3,4,5}{
        \foreach \j in {1,...,10}{
          \node (e\i\j) at (\s * \i,\s * \j) {$r$} ;
          }
      }
    \foreach \i in {2,5}{
      \draw[very thick] (\i * \s + \s/2,\vb) -- (\i * \s + \s/2,\ve) ;
    }
    \foreach \j in {2,6,9}{
      \draw[very thick] (\hb,\j * \s + \s/2) -- (\he,\j * \s + \s/2) ;
    }
   \foreach \j in {4,7}{
      \draw[dashed] (\hb,\j * \s + \s/2) -- (\he,\j * \s + \s/2) ;
    }
    \foreach \i/\j in {1.5/1,10/6}{
      \node at (0.6,\i * \s) {$R_\j$} ;
    }
    \foreach \i/\j/\k in {4.5/2/3,8/4/5}{
      \node at (0.3,\i * \s) {$R_\j \cup R_\k$} ;
    }
    \node at (4 * \s,0) {$C$} ;
    \foreach \i/\j/\l in {e33/e56/z1,e37/e59/z2}{
      \node[draw,fill,fill opacity=0.05,red,dashed,inner sep=-0.03cm,fit=(\i) (\j)] (\l) {} ;
    }
     \end{scope}
     \foreach \i/\j/\z in {x1/y1/15,x2/z1/15,x3/y2/2,x4/z2/10,x5/y3/4}{
       \draw[very thick, -stealth] (\i) to [bend left=\z] (\j) ;
     }
  \end{tikzpicture}
  \caption{Illustration of the injection from the mixed zones and cuts of $C$ on $(M',(\mathcal R,\mathcal C'))$ to the mixed zones of $C$ on $(M',(\mathcal R_1,\mathcal C'))$ and $(M',(\mathcal R_2,\mathcal C'))$. The mixed cuts are represented by red solid lines, and an arbitrary choice of an overlapping $0,1$-corner.}
  \label{fig:injection}
\end{figure}


  Let $\alpha_1$, respectively~$\alpha_2$, be the average, taken among every $C \in \mathcal C'$, of the number of mixed zones on $(M',(\mathcal R_1,\mathcal C'))$, respectively~$(M',(\mathcal R_2,\mathcal C'))$.
  Summing up the last inequality for every $C \in \mathcal C'$, it holds that $\gamma \leqslant \alpha_1+\alpha_2$.
  Thus $\alpha_1+\alpha_2 > 2c_d$.
  Without loss of generality, we assume that $\alpha_1 > c_d$.
  An important point is that $|\mathcal R_1| \leqslant \lceil |\mathcal C|/2 \rceil \leqslant |\mathcal C'|$.
  So by Marcus-Tardos theorem (\cref{thm:marcustardos}) applied to the $0,1$-matrix with as many entries as zones of $(\mathcal R_1,\mathcal C')$, and a 1 in a mixed zone and a 0 otherwise, we obtain a $d$-mixed minor in a coarsening of $(M,(\mathcal R,\mathcal C)) \in \mathcal M_{n,d}$.
  This contradicts \cref{lem:coarsening-mixed-free}, since neatly divided matrices of $M_{n,d}$ are $d$-mixed free.
\end{proof}

Finally we check again, with our slightly different definition of mixed value (compared to that of \cite[Section 5]{twin-width1}), that the column-part fusions can only decrease the mixed value of row-parts (and vice versa).

\begin{lemma}\label{lem:coarsening-mixed-value-decr}
  Let $(M',(\mathcal R,\mathcal C'))$ be the coarsening of a neatly divided matrix $(M,(\mathcal R,\mathcal C))$ resulting from the fusion of a single pair of consecutive parts $C, C' \in \mathcal C$, with $C \cup C'=C^*$.
  Then for every part $R \in \mathcal R$, the mixed value of $R$ on $(M',(\mathcal R,\mathcal C'))$ is at most the mixed value of $R$ on $(M,(\mathcal R,\mathcal C))$.
\end{lemma}
\begin{proof}
  Again this symmetrically works if we switch the role of $\mathcal R$ and $\mathcal C$.
  (The proof of that statement follows as in~\cite[Lemma 11]{twin-width1}.)
  If the zone $R \cap C^*$ is not mixed in $(M',(\mathcal R,\mathcal C'))$, then the mixed value of $R$ has not changed after the fusion of $C$ and $C'$.
  If, on the contrary, the zone $R \cap C^*$ is mixed in $(M',(\mathcal R,\mathcal C'))$, then at least one of the three following propositions holds: zone $R \cap C$ is mixed in $(M,(\mathcal R,\mathcal C))$, zone $R \cap C'$ is mixed in $(M,(\mathcal R,\mathcal C))$, the border between $C$ and $C'$ is a mixed cut for $R$.
  Thus we can charge the contribution of $R \cap C^*$ to the mixed value in $(M',(\mathcal R,\mathcal C'))$ to a unit of mixed value in $(M,(\mathcal R,\mathcal C))$.
  Besides, the borders of $R \cap C^*$ cannot contribute mixed cuts for $R$, since the zone is mixed (recall the definition of a mixed cut for neatly divided $0,1,r$-matrices).
  Finally the remaining mixed zones and mixed cuts of $R$ stayed unchanged between $(M,(\mathcal R,\mathcal C))$ and $(M',(\mathcal R,\mathcal C'))$.
\end{proof}

We are now equipped to find invariant-preserving coarsenings. 

\begin{lemma}\label{lem:coarsening-linear}
  We set $\ell := 2^{4c_d+1}$ and $s := 8 \ell$.
  Every neatly divided matrix $(M,(\mathcal R,\mathcal C)) \in \mathcal M_{n,d}$ has an invariant-preserving coarsening $(M',(\mathcal R',\mathcal C')) \in \mathcal M_{n,d}$ with at least $\lfloor n/s \rfloor$ disjoint pairs of identical columns.
\end{lemma}
\begin{proof}
  We maintain a set $B$ of parts of size at least $2^{4c_d+1}+1$, and refer to these parts as \emph{large}.
  Note that a large part has more than $\ell$ elements, and every part of a neatly divided matrix of $M_{n,d}$ has at most $2 \ell$ elements.
  A part with at most $\ell$ elements is called a \emph{small} part.
  The general plan is to coarsen $(\mathcal R,\mathcal C)$ by successive invariant-preserving symmetric fusions (i.e., elementary coarsening) of pairs of small parts, until $|B| \geqslant \lfloor n/s \rfloor$.
  At that point, we will be able to find a pair of identical columns in each large part.
  The crux of the current lemma is to show that we can always perform a symmetric fusion and remain in the class $M_{n,d}$ (mainly, keep the mixed value below $4c_d$), even when a small fraction of the parts can no longer be merged (mainly, because they are large).

  As an important rule for the fusion, we never merge a large part with another part.
  We set $h := |\mathcal R|=|\mathcal C|$, and greedily find $z := h - 2n/s$ disjoint pairs of small consecutive parts in $\mathcal C$, say, $(C_1,C'_1), \ldots, (C_z,C'_z)$.
  As $n \leqslant 2 \ell h$, it holds that $z \geqslant n/(2 \ell) - 2n/s = 2n/s$.
  We call \emph{frozen} any part of $\mathcal C$ which is not among $(C_1,C'_1), \ldots, (C_z,C'_z)$ (because it is large or next to a large part). 
 
  Let $(\mathcal R,\mathcal C^*)$ be the division resulting from the fusion of the pair of consecutive parts $(C_i,C'_i)$ into say, $C^*_i$, for every $i \in [z]$.
  As $h \leqslant 2|\mathcal C^*|$, the average mixed value among the parts of $\mathcal C^*$ is, by \cref{lem:coarsening-amv}, at most $2c_d$.
  Since $z > 2n/s$ there are more parts $C^*_i$ than frozen parts.
  Hence the average mixed value among the non-frozen parts of $\mathcal C^*$ on $(\mathcal R,\mathcal C^*)$ is at most~$4c_d$.
  This means that there is a merged part $C^*_i$ whose mixed value on $(\mathcal R,\mathcal C^*)$, hence on $(\mathcal R,\mathcal C \cup \{C^*_i\} \setminus \{C_i,C'_i\})$, is at most $4c_d$.
  We perform this fusion.
  Every zone of $C^*_i$ which is mixed is filled with $r$ entries.
  This may come from the fusion of a mixed zone with any other zone, or two zones whose union has a $0,1$-corner. 
  Immediately afterwards we perform the fusion of the corresponding two parts in $\mathcal R$, and the similar update of the entries.
  If $C^*_i$ is large, we add it to $B$.

  Let us show that this elementary fusion (i.e., single symmetric fusion) is invariant-preserving.
  We already established that the mixed value of $C^*_i$ is at most $4c_d$.
  By \cref{lem:coarsening-mixed-value-decr}, the other mixed values have not increased, so they still do not exceed $4c_d$.
  The same applies after the symmetric fusion of two parts of $\mathcal R$.
  After that elementary coarsening, the matrix and the division are still symmetric.
  By \cref{lem:coarsening-mixed-free}, the new neatly divided matrix is still $d$-mixed free.
  Finally because we merged two small parts in $\mathcal C$ and two small parts in $\mathcal R$, still no part exceeds $2 \ell = 2^{4c_d+2}$.
  Hence the new neatly divided matrix is indeed still in $\mathcal M_{n,d}$.

  We proceed with these invariant-preserving elementary fusions until $B$ contains at least $\lfloor n/s \rfloor$ parts.
  Let $(M',(\mathcal R',\mathcal C')) \in \mathcal M_{n,d}$ be the neatly subdivided matrix that we eventually reach.
  We claim that there is a pair of identical column in each part $C$ of $B$.
  Since the mixed value of $C$ on $(\mathcal R',\mathcal C')$ is at most $4c_d$, we claim that the number of different columns is at most~$2^{4c_d+1}=\ell$.
  (This part of the proof follows the second paragraph of the proof of~\cite[Theorem 9]{twin-width1}.)
  Indeed let us consider maximal \emph{blocks} of consecutive (non-mixed) vertical zones $C \cap R_i$ not separating by a mixed cut.
  A block ends at a mixed cut or just before a mixed zone, so there are at most $4c_d+1$ such blocks.
  Observe that a block, seen as a single zone, is vertical (otherwise there would be a $0,1$-corner, hence a mixed cut).
  We also notice that outside of these blocks all the columns of $C$ are equal, since they traverse mixed zones (filled with $r$~entries) and horizontal zones.
  Finally there are only two columns within a block: all 0 entries or all 1 entries.
  Therefore there are at most $2^{4c_d+1}$ pairwise-distinct columns.
  
  By definition of a large part, $|C| \geqslant 2^{4c_d+1}+1$.
  Thus we find two equal columns in~$C$.
\end{proof}

Now it will become apparent why we are filling the mixed zones with $r$ entries.
This allows to simulate a contraction as a simple deletion of an equal row (and a symmetric equal column).
The following lemma is straightforward and states that this operation is invariant-preserving in $\mathcal M_{\cdot,d}$.

\begin{lemma}\label{lem:simple-deletion}
  Let $(M,(\mathcal R,\mathcal C)) \in \mathcal M_{n,d}$ be a neatly divided matrix with two equal rows $\rho, \rho'$ in a part $R \in \mathcal R$, hence symmetrically two equal columns $\gamma, \gamma'$ in a part $C \in \mathcal C$.
  Then removing row $\rho'$ and the symmetric column $\gamma'$ yields a neatly divided matrix of $\mathcal M_{n-1,d}$.
\end{lemma}
\begin{proof}
  By design the new matrix and division are symmetric.
  The new neatly divided matrix remains $d$-mixed free.
  The part size can only decrease, as well as the mixed value.
\end{proof}

\subsection{Bounded twin-width classes have bounded versatile twin-width}

We can now use \cref{lem:coarsening-linear} to find linearly many pairs of vertices that can be contracted, and \cref{lem:simple-deletion} to recurse.
This will be our scheme to find a versatile tree of contractions.

\begin{lemma}\label{lem:versatile-tww}
  Every trigraph of twin-width~$d$ has versatile twin-width at most $4c_{2d+2}2^{4c_{2d+2}+2}$.  
\end{lemma}
\begin{proof}
  Let $G$ be an $n$-vertex graph of twin-width $d$, and let $A := A_{\sigma}(G)$ be its adjacency matrix in an order $\sigma$ compatible with a $d$-sequence of $G$.
  By definition $A$ is $d$-twin-ordered, so by \cref{thm:gridtheorem} it is $2d+2$-mixed free.
  We set $d' := 2d+2$, $\ell := 2^{4c_{d'}+1}$, $s := 8 \ell$, and $D := 4 c_{d'} \cdot 2 \ell = 4c_{d'}2^{4c_{d'}+2}$.
  We initialize $(\mathcal R,\mathcal C)$ to the finest division of $A$, that is, $|\mathcal R|=|\mathcal C|=n$.
  Then $(A,(\mathcal R,\mathcal C))$ is a neatly divided matrix of $\mathcal M_{n,d'}$.
  Indeed the mixed value is 0.

  We apply \cref{lem:coarsening-linear} to $(A,(\mathcal R,\mathcal C))$ and find a coarsening $(A',(\mathcal R',\mathcal C')) \in \mathcal M_{n,d'}$ with $\lfloor n/s \rfloor$ disjoint pairs of identical columns $(\gamma_1,\gamma'_1), \ldots, (\gamma_{\lfloor n/s \rfloor},\gamma'_{\lfloor n/s \rfloor})$.
  These pairs of columns correspond to the pairs of vertices $(a_1,b_1), \ldots, (a_{\lfloor n/s \rfloor},b_{\lfloor n/s \rfloor})$.
  We now argue that, for every $i \in [\lfloor n/s \rfloor]$, the contraction of $a_i$ and $b_i$, resulting in $ab_i$, is $D$-correct.
  First let us justify that it is a $D$-contraction.
  The red degree of $ab_i$ is bounded by the number of red entries of $\gamma_i$ (since we filled the mixed zones with $r$ entries).
  So by \cref{lem:ndm-red-number}, it is bounded by $4c_{d'}2^{4c_{d'}+2} = D$.
  The red degree of the other vertices can increase by one, but again by \cref{lem:ndm-red-number}, it does not exceed $D$.
  The contraction is $D$-correct.
  Indeed applying repeatedly \cref{lem:coarsening-linear} followed by \cref{lem:simple-deletion} gives a sequence of $D$-contractions.
  This stops when ``$\lfloor n/s \rfloor = 0$'', that is $n < s < D$.
  At that point, finishing the contraction sequence in any way builds a complete $D$-sequence.
  Thus every element of $\mathcal M_{\cdot,d'}$ has twin-width $D$.

  Therefore we have found $\lfloor n/s \rfloor$ $D$-correct contractions on disjoint pairs of vertices.
  They constitute the children of the root labeled by $G$ in a versatile tree of $D$-contractions.
  For each $i \in [\lfloor n/s \rfloor]$, by \cref{lem:simple-deletion}, we build the subtree whose root is labeled by $G/a_i,b_i$ with the neatly divided matrix of $\mathcal M_{n-1,d'}$ obtained by removing to $(A',(\mathcal R',\mathcal C'))$ the column $\gamma'_i$ and its symmetric row.
  Thus $G$ has versatile twin-width $D$.
\end{proof}

\subsection{Finishing the proof}
\cref{lem:versatile-tww} is all we need to mimic Norine et al.'s proof for $K_t$-minor free graphs \cite{Norine06}, as described in \cref{subsec:minor-closed}.

\begin{theorem}\label{thm:tww-counting}
  There is a triple-exponential function $f: \mathbb N \rightarrow \mathbb N$ such that the number of $n$-vertex trigraphs with twin-width at most $d$ is at most $n!f(d)^n$.
\end{theorem}

\begin{proof}
  Let $G=(V=[n],E,R)$ be a trigraph with twin-width at most $d$.
  By \cref{lem:versatile-tww}, $G$ has versatile twin-width at most $D := 4c_{2d+1}2^{4c_{2d+2}+2}$, and admits a versatile tree of $D$-contractions.
  We now say that a vertex $u$ is \emph{$D$-good} if there is another vertex $v$ such that the contraction of $u$ and $v$ is $D$-correct.
  The versatile tree of $D$-contractions offers $\lfloor n/s \rfloor$ of $D$-good vertices, with $s := 8 \cdot 2^{4c_{2d+2}+1}= 2^{4(c_{2d+2}+1)}$.
  
  Let $\mathcal I_{n,D}$ be the class of trigraphs with twin-width at most $D$ on vertex set $[n]$ and $\mathcal L_{n,D}$ the subset of $\mathcal I_{n,D}$ consisting of trigraphs such that vertex $n$ is $D$-good.
  Since $\lfloor n/s \rfloor |\mathcal I_{n,D}| \leqslant n|\mathcal L_{n,D}|$, it holds that $|\mathcal I_{n,D}| \leqslant (s+1)|\mathcal L_{n,D}|$.

  Any graph of $\mathcal L_{n,D}$ admits an index $i \in [n-1]$ such that the contraction of vertex $n$ and vertex $i$ is $D$-correct.  
  Therefore any $H \in \mathcal L_{n,D}$ can be obtained from a $H' \in \mathcal I_{n-1,D}$ and $i \in [n-1]$ by splitting $i$ into $i$ and a new vertex $n$, and by linking them to the rest of $H$ observing the following rules.
  Every black edge between $i$ and $j$ in $H'$ forces two black edges $ij$ and $nj$ in $H$.
  Every red edge between $i$ and $j$ in $H'$ forces one of the five alternatives in $H$: a red edge between $i$ and $j$ and anything between $n$ and $j$ (3 alternatives: non-edge, black edge, red edge), a red edge between $n$ and $j$ and a black edge or a non-edge between $i$ and $j$ (2 alternatives).
  Additionally, there might be a non-edge, black edge, or red edge between $i$ and $n$.
  In total, the number of possible graphs $H$ is bounded by $3 \cdot 5^D$.
  Hence $|\mathcal I_{n,t}| \leqslant (s+1) \cdot 3 \cdot 5^D(n-1)|\mathcal I_{n-1,t}| \leqslant 3 \cdot 5^D(s+1)(n-1) \cdot (n-1)!f(d)^{n-1} \leqslant n!f(d)^n$, by setting $f(d) := 3 \cdot 5^D(s+1)=2^{2^{2^{O(d)}}}$.
\end{proof}

\subsection{Showing that a class has unbounded twin-width by counting}\label{sec:application-small}

We have shown that bounded twin-width classes are small.
This may be used to establish that the twin-width of some graphs is unbounded, namely if these graphs do \emph{not} form a small class.
It is not so easy to show that cubic graphs have unbounded twin-width by direct arguments.
\cref{thm:tww-counting} implies this fact by a simple counting argument.
A bipartite cubic graph is the disjoint union of three perfect matchings.
Each matching can be defined in $(n/2)!$ different ways, leading to at least $(n/2)!^3/3^{3n/2}=n^{3n/2+o(n)}$ graphs on vertex set $[n]$, well above $n!c^n = n^{n+o(n)}$.
Similarly, two arbitrary total orders on~$[n]$ can be defined in $(n!)^2$ ways, hence cannot have bounded twin-width.

We will now define a simple class of graphs capturing two arbitrary orders.
Then we will show that these graphs are representable by intervals and by unit disks, and conclude that interval graphs and unit disk graphs have unbounded twin-width.
Of course we did not expect these classes to have bounded twin-width\footnote{In \cite{twin-width1} we show that FO model checking is FPT on bounded twin-width graphs given with a $d$-sequence.}, since FO model checking is W[1]-hard on interval graphs~\cite{MarxS13}, while the mere \textsc{Maximum Independent Set} is W[1]-hard on unit disk graphs~\cite{Marx05}. 
We give a more satisfactory proof of that fact, not using the complexity-theoretic assumption FPT $\neq$ W[1].

We define the (non-hereditary) class $\mathcal B$ by its slices $\mathcal B_n$ of graphs on vertex set $[3n]$.
Each graph of $\mathcal B_n$ has its vertex set partitioned into three cliques of size $n$, say, $(A,B,C)$.
There is no edge between $A$ and $C$.
There are two arbitrary \emph{half-graphs} between $A$ and $B$, and between $B$ and $C$.
To build a half-graph between $A$ and $B$, we first choose an order for the vertices of $A$, say, $a_1, a_2, \ldots, a_n$, and an order for $B$, $b_1, b_2, \ldots, b_n$.
Then we put an edge between $a_i$ and $b_j$ if and only if $i < j$.
The half-graph between $B$ and $C$ is built similarly.
We choose another order for the vertices of $B$, say, $b'_1, b'_2, \ldots, b'_n$, and an order for $C$, $c_1, c_2, \ldots, c_n$.
Then we put an edge between $b'_i$ and $c_j$ if and only if $i < j$.
It is important that the choice of the orders $b_1, \ldots, b_n$ and $b'_1, \ldots, b'_n$ are independent.  

Let us estimate the number of graphs in $\mathcal B_n$, ignoring the single-exponential factors such as the one required to fix the partition $(A,B,C)$.
The half-graph between $A$ and $B$ is defined by choosing a total order for $A$ and a total order for $B$.
There are $n!^2$ such pairs of orders.
Defining the half-graph between $B$ and $C$ requires an additional total order for $B$ (recall that this second ordering of $B$ is independent of its order for the half-graph on $A \cup B$) and a total order for $C$.
Again this amounts to $n!^2$.
Overall there are more than $n!^4$ graphs in $\mathcal B_n$.
Thus $|\mathcal B_n|$ grows like $n^{4n + o(n)}$, while the number of bounded twin-width graphs with vertices labeled by $[3n]$ is only at most $(3n)! c^{3n} = n^{3n + o(n)}$.

One can describe an unlabeled graph of $\mathcal B_n$ with a single permutation $\sigma$ over $[n]$ such that $b'_{\sigma(i)}=b_i$.
\cref{fig:int-unit-disks} shows how to realize a graph of $\mathcal B_n$ as the intersection graph of intervals or as the intersection graph of unit disks, for any given permutation $\sigma$. 

\begin{figure}[h!]
  \centering
  \begin{tikzpicture}
    \def\t{0.2}
    \foreach \b/\e/\j/\c in {1.2/2/1/blue,1.2/2.3/2/blue,1.2/2.6/3/blue,1.2/2.9/4/blue,1.2/3.2/5/blue, 2.1/4.9/1/red,2.4/4/2/red,2.7/5.2/3/red,3/4.6/4/red,3.3/4.3/5/red, 5/6.5/1/black!30!green,4.1/6.5/2/black!30!green,5.3/6.5/3/black!30!green,4.7/6.5/4/black!30!green,4.4/6.5/5/black!30!green}{
      \draw[very thick,color=\c] (\b, \j * \t) -- (\e, \j * \t) ;
    }
    \node at (2,1.25) {$A$} ;
    \node at (3.75,1.25) {$B$} ;
    \node at (5.5,1.25) {$C$} ;

    \begin{scope}[xshift=10cm, yshift=-0.6cm]
      \def\z{0.1}
      \foreach \i in {1,...,5}{
        \draw[fill=blue,fill opacity=0.2] (0,\i * \z) circle (1cm) ;
      }
      \foreach \i in {1,...,5}{
        \draw[fill=black!30!green,fill opacity=0.2] (1.5+\i * \z,2) circle (1cm) ;
      }
      \def\ep{0.02}
      \foreach \j/\i in {1/4, 2/1, 3/5, 4/3, 5/2.3}{
        \draw[fill=red,fill opacity=0.2] (-0.5+\i * \z-\ep,2+\j * \z+\ep) circle (1cm) ;
      }
      \foreach \j/\i in {1/4, 2/1, 3/5, 4/3, 5/2.3}{
        \node[fill,circle,inner sep=-0.015cm] at (-0.5+\i * \z-\ep,2+\j * \z+\ep) {} ;
      }
      \node at (1.15,0.4) {$A$} ;
      \node at (-1.6,2.6) {$B$} ;
      \node at (1.8,3.2) {$C$} ;
    \end{scope}
  \end{tikzpicture}
  \caption{To the left, a representation of a graph of $\mathcal B_5$ by intervals.
    All intervals are obviously stacked up on a single real line, by projection on the $x$-axis.
    To the right, the same graph represented with unit disks.
    The permutation $\sigma$ associated to the graph is 41532.
    In both representations, one can read out the permutation matrix of 41532, where the first row is the bottom one, not the top one.
    For the intervals, this permutation matrix appears in the small gaps between the intervals of $B$ and $C$, while for the unit disks the matrix appears in the centers of the disks of $B$.}
  \label{fig:int-unit-disks}
\end{figure}
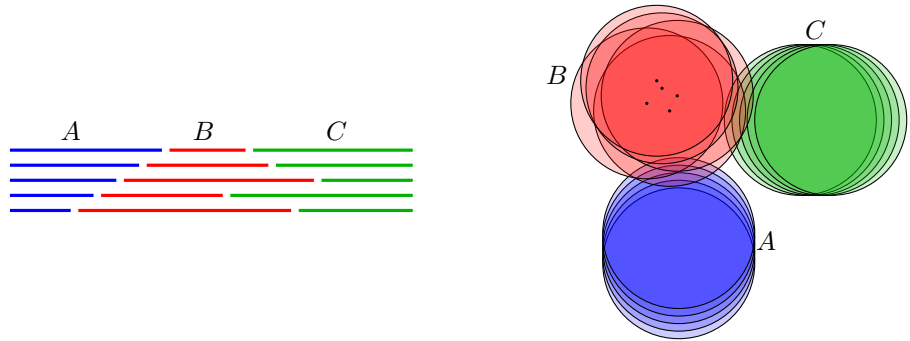

Unit $d$-dimensional ball intersection graphs with bounded clique number have bounded twin-width~\cite{twin-width1}.
One could wonder if $K_t$-free string graphs have bounded twin-width.
\cref{fig:unit-segment} shows that even triangle-free unit segment graphs have unbounded twin-width.
Indeed it shows how to represent any graph of $\mathcal B'_n$ with axis-parallel triangle-free unit segments, where $\mathcal B'_n$ is defined analogously to $\mathcal B_n$ but the sets $A, B, C$ induce now independent sets, and not cliques.
The same argument establishes that the growth of $\mathcal B'_n$ is not the one of a small class.

Let us say that a class $\mathcal C$ is \emph{$t$-bounded} if there is a function $f_{\mathcal C}$ such that every $K_t$-free graph $G$ of $\mathcal C$ have twin-width at most $f_{\mathcal C}(t)$. 
The previous remark shows that there are classes that are $\chi$-bounded but not $t$-bounded, since unit segment graphs are $\chi$-bounded \cite{Suk14}.
In a subsequent paper \cite{twin-width3}, we show that classes of bounded twin-width are $\chi$-bounded.
This implies in particular that every $t$-bounded class is $\chi$-bounded, hence the set of $t$-bounded classes is a proper subset of the set of $\chi$-bounded classes.

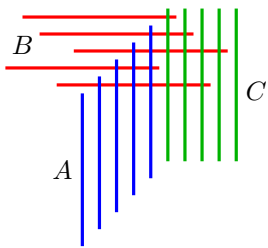
\begin{figure}[h!]
  \centering
  \begin{tikzpicture}[scale=0.75]
    \def\s{0.3}
    \def\h{2.7}
    \foreach \j/\i in {5/2,4/3,3/5,2/1,1/4}{
      \draw[very thick,red] (\s * \i,\s * \j) --++(-\h,0) ;
    }
    \foreach \i in {1,...,5}{
      \draw[very thick,black!30!green] (\s * \i + \s/2,5 * \s + \s/2) --++(0,-\h) ;
    }
    \foreach \i in {1,...,5}{
      \draw[very thick,blue] (\s * \i - 5 * \s + \s/2,\s * \i - \s/2) --++(0,-\h) ;
    }
    \node at (-1.4,-1.2) {$A$} ;
    \node at (-2.1,1) {$B$} ;
    \node at (2,0.2) {$C$} ;
  \end{tikzpicture}
  \caption{A representation of the graph of \cref{fig:int-unit-disks}, where the cliques induced by $A, B, C$ are replaced by independent sets, with axis-parallel triangle-free unit segments. The upside-down permutation matrix of $\sigma = 41532$ is still visible as the right endpoints of the red segments.}
  \label{fig:unit-segment}
\end{figure}

\section{Short parallel $d$-sequences and adjacency labeling schemes}\label{sec:adjacency}

Every $d$-contraction sequence of an $n$-vertex graph has length exactly $n-1$, since each of its steps contracts exactly one pair of vertices.
What if we allow \emph{parallel contractions} where disjoint pairs of vertices may be contracted in a single step?
In this section we adapt the results of \cref{sec:small} on versatile twin-width to prove the existence of parallel contraction sequences of logarithmic length.
We then use them to provide an $f(d) \log n$-adjacency labeling scheme for graphs of twin-width at most~$d$.

A \emph{parallel contraction} in a trigraph $G$ consists of the successive contractions of any number of pairs of vertices $\{a_1,b_1\},\dots,\{a_{\ell},b_{\ell}\}$, where $a_1,\dots,a_\ell,b_1,\dots,b_\ell$ are all distinct.
One can check that the resulting trigraph does not depend on the order in which the pairs are contracted.
Thus instead of the contraction of a \emph{sequence} of pairs, we may as well speak of the parallel contraction of a \emph{set} of disjoint pairs.
A \emph{sequence of parallel $d$-contractions}, or \emph{parallel $d$-sequence} of a trigraph~$G$ is a sequence of $d$-trigraphs $G_k, \ldots, G_1$ where $G_k = G$, $G_1 = K_1$ is the one-vertex (tri)graph, and~$G_{i-1}$ is obtained from~$G_i$ by a parallel contraction (of disjoint pairs of vertices).
It is noteworthy that the existence of a parallel contraction sequence is equivalent to the existence of a (regular) contraction sequence, up to a multiplicative factor in the red degree.

\begin{proposition}\label{lem:par-contr-equiv}
  Let $G$ be a trigraph, and $d \in \mathbb N$.
  \begin{compactitem}
    \item If~$G$ admits a $d$-sequence, then~$G$ also admits a parallel $d$-sequence.
    \item If~$G$ admits a parallel $d$-sequence, then~$G$ also admits a $(2d+1)$-sequence.
  \end{compactitem}
\end{proposition}
\begin{proof}
  The first item is clear since parallel contractions generalize mere contractions.

  We now show the second item.
  Let $G$ and $G'$ be $d$-trigraphs, with~$G'$ obtained from $G$ by the parallel contraction of $\{a_1,b_1\}, \ldots,\{a_\ell,b_\ell\}$.
  This parallel contraction can be sequentialized as $G_0, \ldots, G_\ell$ where $G_0 = G$, and~$G_i$ is obtained from~$G_{i-1}$ by contracting $\{a_i,b_i\}$ into $ab_i$, so that $G_{\ell} = G'$.
  We claim that~$G_i$ is a $(2d+1)$-trigraph for any $i \in [0,\ell]$.

  Consider $x \in V(G_i)$, and let $B^R_{G_i}(x) \subseteq V(G_i)$ be composed of~$x$ and all its red neighbors in~$G_i$.
  There is a natural embedding $e : V(G_i) \to V(G')$ through contraction, namely $e(a_j) = e(b_j) = ab_j$ for $i < j \le \ell$, and $e(x) = x$ for any other vertex.
  By definition of trigraph contractions, if~$xy$ is a red edge in~$G_i$, then either $e(x) = e(y)$, or $e(x)e(y)$ is a red edge in~$G'$.
  Hence $e\left(B^R_{G_i}(x)\right) \subseteq B^R_{G'}\left(e(x)\right)$.
  Furthermore, because $e$ corresponds to the contraction of disjoint pairs, any $X \subseteq V(G_i)$ satisfies $\card{X} \le 2 \card{e(X)}$.
  Finally, we have $\card{B^R_{G'}(e(x))} \le d+1$ because $G'$ is a $d$-trigraph.
  Combining these three claims, we get
  \[ \card{B^R_{G_i}(x)}
    \le 2 \card{e\left(B^R_{G_i}(x)\right)}
    \le 2 \card{B^R_{G'}\left(e(x)\right)}
    \le 2(d+1).
  \]
  Hence the red degree of $x$ in $G_i$ is $\card{B^R_{G_i}(x)}-1 \le 2d+1$.

  Thus if~$G$ and $G'$ are $d$-trigraphs and~$G'$ is obtained from~$G$ by a parallel contraction, then any sequentialization of the parallel contraction produces a sequence of~$(2d+1)$-trigraphs.
  Applying this result to every step of a parallel $d$-sequence yields a $2d+1$-sequence.
\end{proof}

Our main result on parallel contraction sequences is that one can always find a parallel sequence of logarithmic length, at the cost of an increase in the red degree.
This is a variant of the versatile twin-width theorem presented in \cref{sec:small} (\cref{lem:versatile-tww}).
\begin{lemma}\label{lem:short-d-sequence}
  Any $n$-vertex graph $G$ with twin-width at most~$d$ admits a parallel $D$-sequence of length $O(s \cdot \log n)$
  where $s,D$ are double exponential functions of $d$.
\end{lemma}
\begin{proof}
  The proof is very similar to the one of \cref{lem:versatile-tww}.
  Let~$G$ be an $n$-vertex graph with twin-width at most~$d$, and let~$A$ be a $d$-twin-ordered adjacency matrix of $G$.
  By \cref{thm:gridtheorem}, $A$ is $2d+2$-mixed free.
  We set $d' := 2d+2$, $\ell := 2^{4 c_{d'} + 1}$, $s := 8\ell$ and $D := 4 c_{d'} \cdot 2\ell$.
  Applying \cref{lem:coarsening-linear} to the finest division of~$A$ yields a coarsening $(A',(\mathcal{R}',\mathcal{C}')) \in \mathcal{M}_{n,d'}$ with $\lfloor n/s \rfloor$ disjoint pairs of identical columns $(\gamma_1,\gamma'_1), \ldots, (\gamma_{\lfloor n/s \rfloor},\gamma'_{\lfloor n/s \rfloor})$, corresponding to pairs of vertices $(a_1,b_1), \ldots, (a_{\lfloor n/s \rfloor},b_{\lfloor n/s \rfloor})$.

  The difference with \cref{lem:versatile-tww} is that we want to prove that
  the \emph{parallel} contraction of these pairs of vertices is $D$-correct.
  Nonetheless, the arguments remain the same.
  For any~$i$, since $\gamma_i = \gamma'_i$, the contraction of $(a_i,b_i)$ can be done by simply deleting $\gamma'_i$.
  This yields a neat division of the contracted graph.
  \Cref{lem:simple-deletion} readily generalizes to parallel contractions, hence this new division is still in $\mathcal{M}_{\cdot,d'}$.
  By \cref{lem:ndm-red-number}, the red number of this new division is at most $D$.
  This in turn bounds the red degree of the contracted graph (since mixed zones are filled with $r$ entries).
  Hence the parallel $D$-contraction preserves the membership to $\mathcal M_{\cdot,d'}$.
  Applying repeatedly \cref{lem:coarsening-linear,lem:simple-deletion} gives a sequence of parallel $D$-contractions until reaching a graph of size~$n$ with $n < s < D$, at which point the $D$-sequence can be completed in any way.

  This gives a parallel $D$-sequence $G = G_k, \ldots, G_1 = K_1$ for~$G$.
  Furthermore, for $s \le i \le k$, it satisfies $\card{V(G_{i-1})} \le \lceil (1 - \frac{1}{s})\card{V(G_i)}\rceil$.
  It follows that the length of the sequence is $O(s \cdot \log n)$.
\end{proof}

We now use these short parallel contraction sequences to design adjacency labeling schemes for bounded twin-width graphs.

\begin{lemma}\label{lem:parallel-adjacency-scheme}
  For any $d \in \mathbb N$, there exists a function $A : (\{0,1\}^*)^2 \to \{0,1,r_1,\ldots,r_d\}$ such that any trigraph~$G$ with a parallel $d$-sequence of length $k$ has a labeling $\ell : V(G) \to \{0,1\}^*$ satisfying the following:
  \begin{enumerate}
    \item \label{item:label-size} for any $x \in V(G)$,
      $\card{\ell(x)} = k \cdot \lceil 1+(2d+1)\log 3 \rceil$,
    \item \label{item:label-injective} $\ell$ is injective on $V(G)$, 
    \item \label{item:label-color} for any distinct $x,y \in V(G)$,
      \[
        \begin{cases}
          A(\ell(x),\ell(y)) = 0 & \text{if $xy \not\in E(G) \cup R(G)$ (i.e.,~$xy$ is a non-edge)} \\
          A(\ell(x),\ell(y)) = 1 & \text{if $xy \in E(G)$ (i.e.,~$xy$ is a black edge)} \\
          A(\ell(x),\ell(y)) \in \{r_1,\dots,r_d\} & \text{if $xy \in R(G)$ (i.e.,~$xy$ is a red edge).}
        \end{cases}
      \]
    \item \label{item:label-red-unique} for any distinct $x,y,z \in V(G)$,
      if $A(\ell(x),\ell(y)) = r_i$ and $A(\ell(x),\ell(z)) = r_j$, then $i \neq j$.
  \end{enumerate}
  Note that we do not require $A$ to be symmetric: one may have $A(w_1,w_2) = r_j$ and $A(w_2,w_1) = r_{j'}$ with $j \neq j'$.
  In particular, condition~\labelcref{item:label-red-unique} need not properly $d$-color the red edges.
\end{lemma}
\begin{proof}
  We proceed by induction on the length of the parallel $d$-sequence.
  The base case $G = K_1$ is trivial, with the unique label being empty.

  Let $G$ be a trigraph, and let~$G'$ be obtained from~$G$ by parallel contraction of the pairs $\{a_1,b_1\},\ldots,\{a_h,b_h\}$.
  By induction, let us consider a labeling $\ell' : V(G') \to \{0,1\}^*$ for~$G'$ satisfying conditions~\labelcref{item:label-injective,item:label-color,item:label-red-unique}.
  Before defining a labeling on $G$, let us introduce some notations.
  For $i \in [h]$, let $ab_i \in V(G')$ be the vertex obtained from the contraction of $\{a_i,b_i\}$.
  We define two partial functions $p_0, p_1 : V(G') \to V(G)$, corresponding to the predecessors with respect to contraction:
  \[
    \left\{
      \begin{array}{@{}l@{\text{ and }}l@{\quad}l@{}}
          p_0(ab_i) = a_i & p_1(ab_i) = b_i & \text{for $1 \le i \le h$} \\
          p_0(x) = x & \text{$p_1(x)$ is undefined} & \text{for any other $x \in V(G')$.}
      \end{array}
    \right.
  \]
  Note that any $y \in V(G)$ can be uniquely written as $p_c(x)$ for some $x \in V(G')$ and $c \in \{0,1\}$.
  Next, for $x \in V(G')$ and $j \in [d]$, let us define the $j$-th red neighbor of~$x$, denoted by $n_{r_j}(x)$.
  By condition~\labelcref{item:label-red-unique}, there can be at most one $y \in V(G') \setminus \{x\}$ such that $A(\ell'(x),\ell'(y)) = r_j$.
  We define $n_{r_j}(x)$ to be this unique~$y$ if it exists, and to be undefined otherwise.

  Finally for a trigraph~$H$ and any two distinct vertices $x, y \in V(H)$, the color of $xy$ is
  \[ col_H(x,y) =
    \begin{cases}
      1 & \text{if $xy \in E(H)$} \\
      r & \text{if $xy \in R(H)$} \\
      0 & \text{otherwise.}
    \end{cases}
  \]

  We can now define the labeling $\ell : V(G) \to \{0,1\}^*$.
  Given $y \in V(G)$, let $c \in \{0,1\}$, $x \in V(G')$ be such that $y = p_c(x)$.
  Then, $\ell(y)$ consists of the following fields:
  \begin{enumerate}
    \item \label{item:field-recurse} $\ell'(x)$
    \item \label{item:field-id} $c$
    \item \label{item:field-twin-col} $col_G(p_0(x),p_1(x))$
    \item \label{item:field-red-col} For every $j \in [d]$ and $c' \in \{0,1\}$, $col_G(y, p_{c'}(n_{r_j}(x)))$.
  \end{enumerate}

  The fields~\labelcref{item:field-twin-col,item:field-red-col} call partial functions (namely $p_0$, $p_1$, and~$n_{r_j}$).
  If any of these functions is undefined on the relevant values, we use the convention to set the color to~$0$ ($1$ would also be acceptable, but $r$ must be avoided).

  Let us now explain how $A$ is defined to inductively decode these labels.
  Note first that fields~\labelcref{item:field-id,item:field-twin-col,item:field-red-col} have fixed size.
  Thus distinguishing the different fields is not an issue.
  Let $y_1, y_2 \in V(G)$ be two distinct vertices, with $y_1 = p_{c_1}(x_1)$ and $y_2 = p_{c_2}(x_2)$.
  As a first step, we want to retrieve $col_G(y_1,y_2)$ from $\ell(y_1),\ell(y_2)$.
  There are several cases.
  \begin{itemize}
    \item If $x_1 = x_2$, i.e., $y_1$ and $y_2$ are contracted together, then $col_G(y_1,y_2)$ is given by field~\labelcref{item:field-twin-col}.
      Furthermore, we are able to test if $x_1 = x_2$ using~$\ell'(x_1)$ and~$\ell'(x_2)$ (field~\labelcref{item:field-recurse}), since $\ell'$ is injective (condition~\labelcref{item:label-injective}).
    \item Otherwise, if $col_{G'}(x_1,x_2) \in \{0,1\}$, then necessarily $col_G(y_1,y_2) = col_{G'}(x_1,x_2)$ by definition of a trigraph contraction.
      Furthermore we can compute $col_{G'}(x_1,x_2)$ from $\ell'(x_1),\ell'(x_2)$
      since~$\ell'$ correctly encodes the colors in~$G'$ (condition~\labelcref{item:label-color}).
    \item Otherwise, we have $col_{G'}(x_1,x_2) = r$.
      Then let $j \in [d]$ be such that $A(\ell'(x_1),\ell'(x_2)) = r_j$.
      By definition of~$n_{r_j}$, we have $x_2 = n_{r_j}(x_1)$, hence
      \[ col(y_1,y_2) = col(p_{c_1}(x_1), p_{c_2}(n_{r_j}(x_1))) \]
      is given in field~\labelcref{item:field-red-col} of $\ell(y_1)$.
      The position of this information in field~\labelcref{item:field-red-col} is given by~$j$ (obtained from $\ell'(x_1),\ell'(x_2)$ via $A$) and~$c_2$ (field~\labelcref{item:field-id} in $\ell(y_2)$).
  \end{itemize}

  As a second step, when $col(y_1,y_2) = r$, we need to define the \emph{numbered} red label~$r_j$ such that $A(\ell(y_1),\ell(y_2)) = r_j$, with~$j$ unique among the red edges incident to~$y_1$.
  Here we use the fact that all the red edges incident to~$y_1$ appear in fields~\labelcref{item:field-twin-col,item:field-red-col} of~$\ell(y_1)$.
  Thus, given~$\ell(y_1)$, we can enumerate the red edges incident to~$y_1$, and we fix the numbers on red labels according to this enumeration order.
  Since~$G$ has red degree at most $d$ by hypothesis, labels $r_1, \ldots, r_d$ are sufficient (here, it is important that the color of ``undefined'' fields avoids~$r$).
  Therefore conditions~\labelcref{item:label-color} and \labelcref{item:label-red-unique} are maintained.

  The equality $\ell(y_1)=\ell(y_2)$ implies that $c_1 = c_2$, since their field~\labelcref{item:field-id} should match, and that $x_1=x_2$, as $\ell'$ is injective.
  Thus it implies that $y_1 = p_{c_1}(x_1) = p_{c_2}(x_2) = y_2$, hence $\ell$ is injective.
  Finally, let us analyze the size of the labels.
  Field~\labelcref{item:field-id} uses 1 bit.
  Fields~\labelcref{item:field-twin-col,item:field-red-col} contain $2d+1$ colors, with $3$ possible values.
  This can be encoded on $\lceil (2d+1) \log 3 \rceil$ bits.
  Thus, the label sizes for $\ell$ increase by exactly $\lceil 1+(2d+1)\log 3 \rceil$ compared to $\ell'$, and condition~\labelcref{item:label-size} is preserved.
\end{proof}

From \cref{lem:short-d-sequence,lem:parallel-adjacency-scheme}, we immediately conclude the following. 
\begin{theorem}\label{thm:logn-adjacency-scheme}
  The class of graphs with twin-width at most~$d$ admits a $g(d) \log n$-bits adjacency labeling scheme, where $n$ is the number of vertices and $g$ is a double-exponential function.
\end{theorem}

The labeling scheme can in particular be used to encode an $n$-vertex graph of twin-width at most~$d$ on $2^{2^{\gamma(d+1)}} n \log n$ bits, for some constant $\gamma$.
This offers a significant compression over adjacency lists, since cliques for instance have twin-width~0.
Now if the aim is only to globally compress the whole graph, and not to balance the lengths of the vertex labels, there is a simpler encoding with a better dependency in $d$.
It basically consists of ``reading'' the $d$-sequence $G=G_n, \ldots, G_1=K_1$ backwards.
The encoding of $K_1$ is an identifier on $\lceil \log n \rceil$ bits.
Then to go from $G_i$ to $G_{i+1}$, we write $3 \lceil \log n \rceil + 2$ bits corresponding to the ``split vertex'' $w$, in which two vertices $u, v$ vertex $w$ is split, and whether there is a non-edge, a black edge, or a red edge between $u$ and $v$, followed by $d(\lceil \log n \rceil + 4)$ bits corresponding to the edges between $u, v$ and the at most~$d$ vertices adjacent to $w$ in the red graph of~$G_i$.
The latter part is carried by writing down the identifier of each red neighbor $z$ of~$w$ followed by two pairs of bits encoding if there is a non-edge, a black edge, or a red edge between $u$ and $z$, and between $v$ and $z$. 
This permits to reconstruct $G$, and store it on only $(d+3) n \lceil \log n \rceil + (4d+2)n$ bits.

\section{Expanders with bounded twin-width}\label{sec:expanders}

A \emph{2-lift} of a graph $G$ is a graph $G'$ on twice as many vertices, built by duplicating every vertex $v \in V(G)$ into two copies, say, $v_1$ and $v_2$, and for every edge $vw \in E(G)$, adding to $E(G')$ either the edges $v_1w_1$ and $v_2w_2$ (parallel) or the edges $v_1w_2$ and $v_2w_1$ (crossing).
The choice, for each edge of $G$, of having two \emph{parallel} edges or two \emph{crossing} edges is called the \emph{signing} of the edges.
See~\cref{fig:lift} for an example of a 2-lift.
Observe that $G$ has $2^{|E(G)|}$ possible 2-lifts or signings.
For instance, the all-parallel signing gives two disjoint copies of $G$, while the all-crossing signing gives the bipartite adjacency graph of $G$.

\begin{figure}[h!]
  \centering
  \begin{tikzpicture}
    \def\s{1.5}
    \foreach \i/\j/\l in {0/0/a,0/1/b,1/0/c,1/1/d}{
      \node[draw,circle] (\l) at (\i * \s,\j * \s) {} ;
    }
    \foreach \i/\j in {a/b,a/c,a/d,b/c,b/d,c/d}{
      \draw (\i) -- (\j) ;
    }

    \draw[thick,-stealth] (1.75 * \s,\s /2) --++(1,0) ;

    \begin{scope}[xshift=5cm]
     \foreach \i/\j/\l in {-0.2/0/a1,0/-0.2/a2,-0.2/1/b1,0/1.2/b2,1/1.2/c1,1.2/1/c2,1.2/0/d1,1/-0.2/d2}{
       \node[draw,circle] (\l) at (\i * \s,\j * \s) {} ;
    }
    \foreach \i/\j in {a1/b1,a2/b2, b1/c1,b2/c2, a1/c1,a2/c2, d1/a2,d2/a1, d1/b1,d2/b2, d1/c2,d2/c1}{
      \draw (\i) -- (\j) ;
    }
    \end{scope}
  \end{tikzpicture}
  \caption{An example of a 2-lift of $K_4$.}
  \label{fig:lift}
\end{figure}
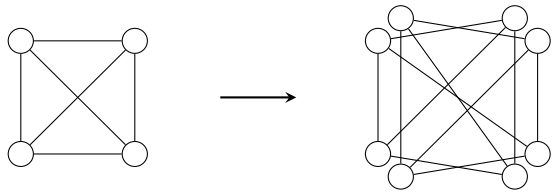

For $n$ a power of 2, performing a sequence of $\log n - 2$ randomly-signed 2-lifts starting on $K_4$ yields an $n$-vertex expander almost surely~\cite{Bilu06}.
Observe that the obtained graph is necessary cubic since the 2-lift operation preserves the degree.
Bilu and Linial~\cite{Bilu06} even exhibit a deterministic polytime procedure to actually find the signings leading from $K_4$ to a cubic expander.
The next result shows that cubic expanders can have bounded twin-width.

\begin{lemma}\label{lem:BiluLinial}
Every graph obtained from $K_4$ by performing a sequence of 2-lifts has twin-width at most~6.
\end{lemma}

\begin{proof}
  We show that if $G$ is a cubic graph and $G'$ is a 2-lift of $G$, then $G$ can be obtained from $G'$ by a sequence of contractions in which the maximum degree never goes above 6.
  It is enough to conclude since $K_4$ is obviously 6-collapsible, and we can assume that the cubic trigraph we start from has all its edges red.

  Let $v^1,v^2, \ldots, v^n$ be the vertices of $G$, and $v^i_1, v^i_2$ be the duplicates of $v^i$ in $G'$.
  For each $i$ running from 1 to $n$, we contract $v^i_1$ and $v^i_2$.
  By definition of a 2-lift, after these $n$ contractions, the graph obtained is $G$.
  We contracted disjoint pairs of vertices of degree 3, so we could not create vertices of degree more than 6. 
\end{proof}

This surprising result teaches us the following lessons.
First, bounded twin-width appears more general than expected.
Also, by \cref{thm:transduction}, not only there are some expanders with bounded twin-width but there are some FO transductions of expanders with that property.  
Second, it tells us that even among bounded-degree graphs, bounded twin-width is a novel class.
Indeed bounded twin-width could have coincided with polynomial expansion within the class of bounded-degree graphs.
Now we know that it is not the case.
There are cubic graphs with bounded twin-width but no strongly sublinear (i.e., of size at most $n^{1-\varepsilon}$ for some $\varepsilon > 0$) balanced separators.
Expanders have treewidth $\Theta(n)$ and therefore no strongly sublinear balanced separators, the latter being equivalent to polynomial expansion \cite{Plotkin94,Dvorak16}.

The third lesson is that designing good approximation algorithms in bounded twin-width classes promises to be challenging.
It is perfectly fitting and propitious to ask for other algorithmic applications of twin-width.
Before we understand enough to approximate in general bounded twin-width classes, an interesting first step is to approximate optimization problems such as \textsc{Maximum Independent Set} (\textsc{MIS} for short) on graphs with bounded degree and twin-width.
\textsc{MIS} is APX-hard in general cubic graphs, so we may ask for a polynomial-time approximation scheme (PTAS) when we add the condition of bounded twin-width.
A natural approach for that would be to show that these graphs have strongly sublinear balanced separators (this is how PTASes are obtained for planar, $H$-minor free graphs, etc.).
This approach is now ruled out.
Therefore, if \textsc{MIS} indeed admits a PTAS in bounded twin-width cubic graphs, this cannot be directly based on small balanced separators.
The simplest toy-problem in that direction is to explore PTASes for iterated 2-lifts of $K_4$.

\section{Subdivisions of cliques}\label{sec:subd-cliques}


For any non-negative integer $k$, the \emph{$k$-subdivision} of a graph~$G$, denoted by $G^{(k)}$, is the graph obtained by subdividing every edge of $G$ exactly~$k$ times.
For any $f : \mathbb N \to \mathbb N$, let $\mathcal G_f$ be the class formed by the $f(|V(G)|)$-subdivision of every graph $G$.

\begin{theorem}\label{thm:clique-subdivision}
  For every positive and non-decreasing $f$, $\mathcal G_f$ has bounded twin-width if and only if $f(n) = \Omega(\log n)$.
\end{theorem}

Let us first observe that for any integer $k > 0$ and $n$-vertex graph $G$, $G^{(k)}$ is an induced subgraph of $K_n^{(k)}$.
Thus the class $\mathcal G_f$ is contained in the hereditary closure of the graphs $K_n^{(f(n))}$ for $n \ge 0$.
Since twin-width never increases when taking induced subgraphs, it suffices to consider graphs of the form $K_n^{(f(n))}$.
As hinted at in \cref{sec:prelim:small}, the forward implication of \cref{thm:clique-subdivision} could be derived from \cref{thm:tww-counting} and the fact that $o(\log n)$-subdivisions does \emph{not} form a small class.
We give a direct proof of a stronger statement.

\begin{proposition}
  For $d \ge 0$ and $k > 0$ integers, if $K_n^{(k)}$ has twin-width at most~$d$, then $k \ge \log_{d+1} (n-1) - 1$.
\end{proposition}
\begin{proof}
  Let $G$ be $K_n^{(k)}$, for some positive integer $k$.
  Assuming that $G$ has twin-width at most~$d$, we show that $k \ge \log_{d+1} (n-1) - 1$.
  Note that the assumption $k>0$ is required because $K_n^{(0)}=K_n$ has twin-width~0.

  In a $d$-contraction sequence of $G$, let us consider the first step in which two vertices $x, y$ of the original $K_n$ are contracted.
  Let $\mathcal P$ the partition of $V(G)$ at this step, and $P_0 \in \mathcal P$ the part containing $x$ and $y$.
  In $G$, consider the $n-1$ paths, on $k+1$ edges each, resulting from the subdivided edges starting at $x$.
  We partition the vertices of these paths as $V_1, \ldots, V_{k+1}$, where $V_i$ contains all the vertices at distance $i$ of $x$.
  Then $V_{k+1}$ contains all the vertices of the original $K_n$ except $x$.
  In particular, no two vertices of $V_{k+1}$ are in the same part of $\mathcal P$.

  All the vertices of $V_1$ are neighbors of $x$ but not of $y$, thus for any part $P \in \mathcal P \setminus \{P_0\}$ intersecting $V_1$, $P P_0$ is a red edge in $G_{\mathcal P}$.
  Thus at most $d+1$ parts of $\mathcal P$ intersect $V_1$, and there exists $P_1 \in \mathcal P$ such that $|P_1 \cap V_1| \ge \frac{n-1}{d+1}$.
  Observe that $P_1$ may well be equal to $P_0$.
  Similarly the vertices in $P_1 \cap V_1$ have pairwise-disjoint neighborhoods in $V_2$, hence $V_2 \cap N_G(P_1 \cap V_1)$, of size at least $\frac{n-1}{d+1}$, is split in at most $d+1$ parts in $\mathcal P$.
  Thus there is a part $P_2 \in \mathcal P$ (that may be $P_0$ or $P_1$) which contains at least $\frac{n-1}{(d+1)^2}$ vertices of $V_2$.
  It follows by induction that for every $i \in [k+1]$, there exists a part of $\mathcal P$ containing at least $\frac{n-1}{(d+1)^i}$ vertices of $V_i$.
  However no part of $\mathcal P$ contains more than one vertex of $V_{k+1}$.
  Hence $\frac{n-1}{(d+1)^{k+1}} \le 1$, and $k \ge \log_{d+1} (n-1) - 1$.
\end{proof}

The converses relies on some results on decompositions of permutations.
We now encode a permutation $\sigma$ in the usual way, as the sparse matrix with entry 1 at position $(i,\sigma(i))$, and 0 elsewhere.
(This is unlike the more cumbersome but technically-motivated dense encodings used in \cref{sec:application-small} and~\cite[Section~6.1]{twin-width1}.)

A permutation $\sigma$ is a \emph{$t$-merge} if its domain
 can be partitioned into $t$ possibly-empty discrete intervals $I_1, \ldots, I_t$ such that the restriction of $\sigma$ to $I_i$ is increasing.
Merging $t$ sorted lists can be expressed as the application of some well chosen $t$-merge to the concatenation of the lists.
A permutation $\sigma$ is a \emph{parallel $t$-merge} if its domain can be partitioned into an arbitrary number of intervals $J_1, \ldots, J_r$ such that $\sigma$ operates independently on each $J_i$ (i.e., $\sigma(J_i) = J_i$), and the restriction $\sigma_{|J_i}$ is a $t$-merge.
See \cref{fig:parallel-merge} for an example of a parallel 2-merge.

\begin{figure}[htb]
  \centering
  \begin{tikzpicture}
    \def\s{0.5}
    \def\hb{\s/2}
    \def\vb{-\s/2}
    \def\he{8.5 * \s}
    \def\ve{-8.5 * \s}

    \begin{scope}[yshift=9 * \s cm]
    \draw (\hb-0.05,\vb) -- (\hb-0.05,\ve) --++(0.1,0) ;
    \draw (\hb-0.05,\vb) --++(0.1,0) ;
    \draw (\he+0.05,\vb) -- (\he+0.05,\ve) --++(-0.1,0) ;
    \draw (\he+0.05,\vb) --++(-0.1,0) ;
    \end{scope}
    
    \begin{scope}[rotate=90]
    \foreach \i/\v in {1/2,2/3,3/5,4/1,5/4}{
      \foreach \j in {1,...,5}{
        \ifnum \j=\v
          \node (e\i\j) at (\s * \i,-\s * \j) {1};
        \else
          \node (e\i\j) at (\s * \i,-\s * \j) {0};
        \fi
      }
    }
    \foreach \i/\j/\v in {6/6/1,7/7/0,7/8/1,8/7/1,8/8/0}{
      \node (e\i\j) at (\s * \i,-\s * \j) {\v};
    }

    \draw (5.5 * \s,\vb) -- (5.5 * \s,-6.5 * \s);
    \draw (\hb,-5.5 * \s) -- (6.5 * \s,-5.5 * \s);
    \draw (6.5 * \s,-5.5 * \s) -- (6.5 * \s,\ve);
    \draw (5.5 * \s,-6.5 * \s) -- (\he,-6.5 * \s);

    \draw[dashed] (3.5 * \s,\vb) -- (3.5 * \s,-5.5 * \s);
    \draw[dashed] (7.5 * \s,-6.5 * \s) -- (7.5 * \s,\ve);
    \end{scope}
  \end{tikzpicture}
  \caption{%
    A parallel 2-merge matrix, corresponding to the permutation 23514687.
    Note that the first row is at the bottom, as is common with permutation matrices.
    It is composed of three blocks, each of which can be partitioned in two increasing subsequences, indicated by the dashes.
    Empty areas are filled with 0.
  }
  \label{fig:parallel-merge}
\end{figure}

\begin{lemma}\label{lem:t-merge-decomposition}
  For any $t,\ell \in \mathbb{N}$, any permutation on $t^\ell$ elements can be decomposed as a product of at most $\ell$ parallel $t$-merges.
\end{lemma}
\begin{proof}
  The case $t = 2$ corresponds to a merge sort.
  In the recursion tree of a merge sort, each level of inductive calls can be expressed as a single parallel 2-merge.
  To sort up to $2^\ell$ elements, a merge sort with recursion depth limited to $\ell$ suffices, and this can be expressed as the composition of $\ell$ parallel 2-merges.
  This generalizes easily to $t$-merges, and composing $\ell$~parallel $t$-merges allows to sort up to $t^\ell$ elements.
  \begin{figure}[htb]
  \centering
  \begin{tikzpicture}
    \def\s{0.5}
    \def\hb{\s/2}
    \def\vb{-\s/2}
    \def\he{8.5 * \s}
    \def\ve{-8.5 * \s}
    \def\z{4.9}

    \foreach \i in {0,\z,2 * \z}{
    \begin{scope}[yshift=9 * \s cm,xshift=\i cm]
    \draw (\hb-0.05,\vb) -- (\hb-0.05,\ve) --++(0.1,0) ;
    \draw (\hb-0.05,\vb) --++(0.1,0) ;
    \draw (\he+0.05,\vb) -- (\he+0.05,\ve) --++(-0.1,0) ;
    \draw (\he+0.05,\vb) --++(-0.1,0) ;
    \end{scope}
    }
    
    \begin{scope}[rotate=90]
    \foreach \i/\v in {1/5,2/4,3/6,4/1,5/3,6/2,7/8,8/7}{
      \foreach \j in {1,...,8}{
        \ifnum \j=\v
          \node (e\i\j) at (\s * \i,-\s * \j) {1};
        \else
          \node (e\i\j) at (\s * \i,-\s * \j) {0};
        \fi
      }
    }
    \end{scope}

     \begin{scope}[rotate=90,yshift=-\z cm]
    \foreach \i/\v in {1/1,2/4,3/5,4/6,6/3,5/2,7/7,8/8}{
      \foreach \j in {1,...,8}{
        \ifnum \j=\v
          \node (e\i\j) at (\s * \i,-\s * \j) {1};
        \else
          \node (e\i\j) at (\s * \i,-\s * \j) {0};
        \fi
      }
    }
     \draw[dashed] (4.5 * \s,\vb) -- (4.5 * \s,\ve);
     \end{scope}

     \begin{scope}[rotate=90,yshift=-2 * \z cm]
    \foreach \v/\i in {1/4,2/2,3/1,4/3}{
      \foreach \j in {1,...,4}{
        \ifnum \j=\v
          \node (e\i\j) at (\s * \i,-\s * \j) {1};
        \else
          \node (e\i\j) at (\s * \i,-\s * \j) {0};
        \fi
      }
    }
    \foreach \v/\i in {5/6,6/5,7/8,8/7}{
      \foreach \j in {5,...,8}{
        \ifnum \j=\v
          \node (e\i\j) at (\s * \i,-\s * \j) {1};
        \else
          \node (e\i\j) at (\s * \i,-\s * \j) {0};
        \fi
      }
    }
    \end{scope}
  \end{tikzpicture}
  \caption{Left: the permutation $\tau$ to sort. Center: the 2-merge permutation $\sigma$ to use on $\tau$. Right: the composition $\sigma^{-1} \circ \tau$, one may inductively sort the two blocks by applying further parallel 2-merges.}
  \label{fig:}
\end{figure}
\end{proof}

The previous lemma is reminiscent of the theory of sorting networks, in that we decompose arbitrary permutations as a product of few base permutations---in our case parallel $t$-merges.
However, sorting networks consider more restricted base permutations (e.g., separable permutations), whereas we merely need the base permutations to have bounded twin-width.

\begin{lemma}\label{lem:t-merge-grid}
 Matrices of parallel $t$-merges are $(t+1)$-grid free.
\end{lemma}
\begin{proof}
  Let $\sigma$ a parallel $t$-merge, with its domain partitioned into intervals $J_1, \ldots, J_r$ such that~$\sigma(J_i)=J_i$, and every $\sigma_{|J_i}$ is a $t$-merge.
  Assume for a contradiction that $\sigma$ contains a $(t+1)$-grid.
  Then it contains a decreasing subsequence of length $t+1$.

  For any $i<j$, $x \in J_i$ and $y \in J_j$, one has $x<y$ and $\sigma(x) < \sigma(y)$ because $J_i,J_j$ are disjoint intervals, with $\sigma(J_i) = J_i$ and $\sigma(J_j) = J_j$.
  It follows that any decreasing subsequence is contained entirely in one of the $J_k$.
  Thus, there exist a $t$-merge $\sigma_{|J_k}$ which contains a decreasing subsequence of length $t+1$.

  Since $\sigma_{|J_k}$ is a $t$-merge, $J_k$ is itself partitioned into intervals $I_1,\dots,I_t$ such that $\sigma$ is increasing on $I_i$.
  Hence each $I_i$ can contain at most one element of a decreasing subsequence, and $\sigma_{|J_k}$ contains no decreasing subsequence of length more than $t$, a contradiction.
\end{proof}

\begin{proposition}\label{prop:bdtww-subd}
  For any $c > 0$, the class of cliques $K_n$ subdivided at least $\frac{\log n}{c}$ times has twin-width at most~$f(c)$ for some triple-exponential function $f$.
\end{proposition}
\begin{proof}
  Let $k \ge \frac{\log n}{c}$, and let $G$ be $K_n^{(k)}$.
  We want to order $V(G)$ such that the adjacency matrix of $G$ in that order is $r$-grid free, for some $r$ depending only on $c$.
  This implies the desired twin-width bound by \cref{thm:gridtheorem}.

  Choose an arbitrary orientation of the edges of $K_n$.
  In $G$, the edges of $K_n$ become directed paths on $k+1$ edges.
  Then, for $0 \le i \le k$, let $V_i \subset V(G)$ contain every $i$-th vertex along these directed paths.
  In particular, $V_0$ corresponds to the vertices of $K_n$, while $V_1, \ldots, V_k$ are all the vertices created by the subdivision.
  Thus, $V_0, \ldots, V_k$ is a partition of $V(G)$.

  Let us now define an order within each $V_i$.
  Choose $x_1, \ldots, x_n$ an arbitrary order on $V_0$.
  The extremal set $V_1$ is ordered according to the neighbors in $V_0$, i.e., with first the neighbors of $x_1$ in any order, then the neighbors of $x_2$, etc.
  We proceed similarly for $V_k$.
  The disjoint paths in $G - V_0$ define a bijection between $V_1$ and $V_k$,
  which can be interpreted as a permutation~$\sigma$ on $\frac{n(n-1)}{2}$ elements according to the previous orderings.
  Then, choosing orderings for $V_2, \ldots, V_{k-1}$ is equivalent to decomposing $\sigma$ as a product $\sigma_1 \circ \cdots \circ \sigma_{k-1}$.
  By \cref{lem:t-merge-decomposition}, we may choose $\sigma_1, \dots, \sigma_{k-1}$ to be parallel $t$-merges for any $t$ such that $t^{\log(n) / c} \ge \frac{n(n-1)}{2}$.
  This is satisfied by $t = \lceil 2^{2c} \rceil$, which crucially is independent of $n$.
  With this choice of decomposition for $\sigma$, we have ordered $V_2, \ldots, V_{k-1}$.
  Finally, $V(G)$ is ordered as $V_0 < \cdots < V_k$, where $V_i$ is ordered as previously defined.

  Let $M$ be the adjacency matrix of~$G$ respecting this ordering.
  Let $R_0, \ldots, R_k$ (resp.\ $C_0,\dots,C_k$) the partition of the rows (resp.~columns) of~$M$ induced by the partition $V_0, \ldots, V_k$ of $V(G)$.
  Then $(\mathcal{R},\mathcal{C}) = (\{R_0,\dots,R_k\},\{C_0,\dots,C_k\})$ is a division of~$M$.
  For $i,j \in [0,k]$, let~$M_{i,j}$ be the zone $R_i \cap C_j$, which corresponds to the adjacency matrix between $V_i$ and $V_j$.
  The zone~$M_{i,j}$ is non-zero if and only if $i = j \pm 1$ modulo $k+1$.
  Thus, there are $2k+2$ non-zero zones, forming a double diagonal with corners (see~\cref{fig:Kn-subdivision-matrix}).
  \begin{figure}[ht]
    \centering
    \begin{tikzpicture}
      \def\s{1.1}
      \def\t{0.8}
      \def\n{5}

      \begin{scope}[rotate=90]
      \node (M10) at (\s/2, \t/2) {$M_{1,0}$};
      \draw (0,0) --+(\s,0) --+(\s,\t) --+(0,\t) -- cycle;
      \node (M01) at (-\t/2, -\s/2) {$M_{0,1}$};
      \draw (0,0) --+(-\t,0) --+(-\t,-\s) --+(0,-\s) -- cycle;

      \node (M\n0) at (\n*\s - \s/2, \t/2) {$M_{k,0}$};
      \draw (\n*\s,0) --+(-\s,0) --+(-\s,\t) --+(0,\t) -- cycle;
      \node (M0\n) at (-\t/2, -\n*\s + \s/2) {$M_{0,k}$};
      \draw (0,-\n*\s) --+(-\t,0) --+(-\t,\s) --+(0,\s) -- cycle;

      \foreach \i/\j in {1/2,2/1,2/3,3/2,5/4,4/5}{
        \draw (\s*\i, -\s*\j) --+(-\s,0) --+(-\s,\s) --+(0,\s) -- cycle;
      }
      \foreach \i/\j in {1/2,2/1,2/3,3/2}{
        \node (M\i\j) at (\s*\i - \s/2, -\s*\j + \s/2) {$M_{\i,\j}$};
      }

      \node (M34) at (\s*3 - \s/2, -\s*4 + 0.6*\s) {$\Ddots$};
      \node (M43) at (\s*4 - \s/2, -\s*3 + 0.6*\s) {$\Ddots$};
      \node (M45) at (\s*4 - \s/2, -\s*5 + \s/2) {$M_{k-1,k}$};
      \node (M54) at (\s*5 - \s/2, -\s*4 + \s/2) {$M_{k,k-1}$};

      \draw[very thick] (-\t,\t) -- (\n*\s,\t) -- (\n*\s,-\n*\s) -- (-\t,-\n*\s) -- cycle;

      \node (R0) at (-\t*1.4, \t/2) {$C_0$};
      \node (C0) at (-\t/2, \t*1.4) {$R_0$};
      \foreach \i in {1,...,3}{
        \node (R\i) at (-\t*1.4, -\i*\s + \s/2) {$C_{\i}$};
        \node (C\i) at (\i*\s - \s/2,\t*1.4) {$R_{\i}$};
      };
      \node (R4) at (-\t*1.4, -4*\s + 0.6*\s) {$\cdots$};
      \node (C4) at (4*\s - \s/2,\t*1.4) {$\vdots$};
      \node (R5) at (-\t*1.4, -5*\s + \s/2) {$C_k$};
      \node (C5) at (5*\s - \s/2,\t*1.4) {$R_k$};
      \end{scope}
    \end{tikzpicture}
    \caption{The adjacency matrix $M$ of $G$, with the appropriate ordering of the vertices.}
    \label{fig:Kn-subdivision-matrix}
  \end{figure}
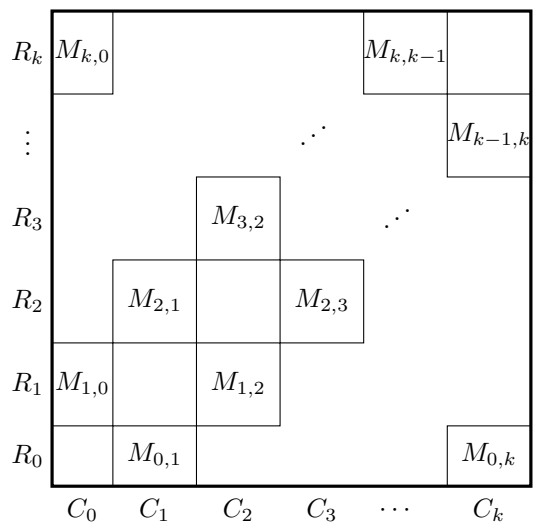

  \begin{claim}\label{clm:zone-grid}
    Every zone of the division $(\mathcal{R},\mathcal{C})$ of~$M$ is $(t+1)$-grid free.
  \end{claim}
  \begin{claimproof}
    For $1 \le i < k$, the zones $M_{i,i+1}$ and $M_{i+1,i}$ are parallel $t$-merges or transposes thereof, hence are $t+1$-grid free by \cref{lem:t-merge-grid}.
    The zones $M_{0,1}$, $M_{1,0}$, $M_{0,k}$, and $M_{k,0}$ are composed of a single monotone sequence, hence are 2-grid free.
  \end{claimproof}

  Let us now consider an $\ell$-grid minor $(\mathcal{R}',\mathcal{C}') = (\{R'_1,\dots,R'_\ell\},\{C'_1,\dots,C'_\ell\})$ of~$M$, i.e., every zone $R'_i \cap C'_j$ contains at least one entry~1.

  \begin{claim}\label{clm:block-intersect-bound}
    There is a set $A \subset \mathcal{C}$ of at most~5 column-parts such that every $C' \in \mathcal C'$ satisfies~$C' \cap \bigcup A \neq \emptyset$.
  \end{claim}
  \begin{claimproof}
    Let $i$ be minimal such that $R'_1 \subseteq R_0 \cup \dots \cup R_i$.
    Then for $i+1 < j < k$, one may verify from the structure of $M$ that $R'_1 \cap C_j$ is full 0.
    Thus any $C' \in \mathcal C'$ must intersect one of $C_0, \ldots, C_{i+1}$ or $C_k$.
    Symmetrically, let $i'$ be maximal such that $R'_\ell \subseteq R_{i'} \cup \dots \cup R_k$.
    Then any $C' \in \mathcal{C}'$ must intersect one of $C_0$ or $C_{i'-1},\dots,C_k$.
    Define $A := \{C_{i'-1}, \dots, C_{i+1}\} \cup \{C_0, C_k\}$.
    The above implies that any $C' \in \mathcal{C'}$ must intersect some $C \in A$.
    Finally we have $i \le i'$, which implies $\card{A} \le 5$.
  \end{claimproof}

  \begin{claim}\label{clm:block-subset-bound}
    There exists $C \in \mathcal{C}$ such that at least $\frac{\ell-10}{5}$ parts of $\mathcal{C}'$ are subsets of $C$.
  \end{claim}
  \begin{claimproof}
    Consider an arbitrary $C' \in \mathcal{C}'$.
    By \cref{clm:block-intersect-bound}, there is some $C \in A$ such that $C' \cap C \neq \emptyset$.
    We consider two cases, depending on whether or not $C' \subseteq C$:
    \begin{itemize}
      \item If $C' \not\subseteq C$, it means that $C'$ contains one of the two boundaries of $C$.
        For a given $C \in A$, there can only be two $C' \in \mathcal{C}'$ for which it is the case.
        Thus this case applies to at most $2\card{A} \le 10$ elements of $\mathcal{C}'$.
      \item Otherwise---and this applies to at least $\ell-10$ elements of $\mathcal{C}'$---we have $C' \subseteq C$ for some $C \in A$.
        Since $\card{A} \le 5$, by pigeonhole principle, there exist $C \in A$ such that at least $\frac{\ell-10}{5}$ elements of $\mathcal{C}'$ are subsets of $C$.
        This proves the claim.
    \end{itemize}
  \end{claimproof}

  Of course, \cref{clm:block-intersect-bound,clm:block-subset-bound} still hold when inverting the roles of rows and columns.
  Thus, there are $R \in \mathcal{R}, C \in \mathcal{C}$ such that $R$ (resp.~$C$) contains at least $\frac{\ell-10}{5}$ parts of $\mathcal{R}'$ (resp~$\mathcal{C}'$) as subsets.
  Hence the zone $R \cap C$ contains an $\frac{\ell-10}{5}$-grid induced by the corresponding parts of $\mathcal R'$ and $\mathcal C'$.
  By \cref{clm:zone-grid}, it follows that $\frac{\ell-10}{5} \le t$, or $\ell \le 5t + 10$.
  Recall that $t$ was chosen as $t = \lceil 2^{2c} \rceil$.
  Hence we have proved that $M$ is $g(c)$-grid free for $g(c) = 5 \lceil 2^{2c} \rceil + 11$.

  A fortiori $M$ is $g(c)$-mixed free, and by \cref{thm:gridtheorem} the twin-width of $G$ is at most $f(c)$
  for some $f(c)$ double-exponential in $g(c)$, hence triple-exponential in $c$.
\end{proof}

In the next section, we will show that graphs with queue number $t$ have twin-width $2^{2^{O(t)}}$ (see \cref{thm:quack}).
This can be used to get an alternative proof to \cref{prop:bdtww-subd}, albeit not self-contained.
Indeed it was shown that the $2\lceil \log_d{\lfloor n/2 \rfloor} \rceil+1$-subdivision of $K_n$ (see \cite[Theorem 4]{DujmovicW05}) has queue number at most~$d$.

\section{Sparse twin-width}\label{sec:sparse-tww}

We start this section by showing the list of equivalences of \cref{thm:sparseboundedtww}.

\subsection{Characterizations}\label{sec:equiv}

We recall that $A_\sigma(G)$ is the adjacency matrix of $G$ when $V(G)$ is ordered by $\sigma$, and that a class $\mathcal C$ is said $d$-grid free if for every $G \in \mathcal C$ there is an ordering $\sigma$ of $V(G)$ such that $A_\sigma(G)$ is $d$-grid free.
\sparse*

\begin{proof}
  We start by showing that $(i)$ and $(ii)$ are equivalent.
  Then we will show that both $(iii)$ and $(iv)$ are implied by $(ii)$, and imply $(i)$.
  
  $(i) \Rightarrow (ii)$.
  Assume that $\mathcal C$ is $K_{t,t}$-free, for some integer $t$.
  Let $A$ be a $d'$-twin-ordered adjacency matrix of $G \in \mathcal C$, where $d' = \tww(G)$.
  By~\cref{thm:gridtheorem}, $A$ is $2d'+2$-mixed free.
  Let $(\mathcal R,\mathcal C) := (\{R_1, \ldots, R_d\}, \{C_1, \ldots, C_d\})$ be a $d$-grid minor of $A$, i.e., such that every zone $R_i \cap C_j$ contains a 1.
  To conclude, we will upper bound~$d$ by $(2d'+2)t-1$.
  For the sake of contradiction, suppose that $d = (2d'+2)t$.
  Let $(\mathcal R',\mathcal C') := (\{R'_1, \ldots, R'_{2d'+2}\}, \{C'_1, \ldots, C'_{2d'+2}\})$ be the division obtained by merging groups of $t$ consecutive parts of $(\mathcal R,\mathcal C)$.
  Since $A$ is $2d'+2$-mixed free, there is a zone $R'_i \cap C'_j$ of $(\mathcal R',\mathcal C')$ which is horizontal or vertical.
  Without loss of generality, let us assume that $R'_i \cap C'_j$ is horizontal.
  Now consider the $(t,t)$-division $(\mathcal R^*,\mathcal C^*) = (\{R_{(i-1)t+1}, \ldots, R_{it}\}, \{C_{(j-1)t+1}, \ldots, C_{jt}\})$ induced by $(\mathcal R,\mathcal C)$ on $R'_i \cap C'_j$.
  Since $(\mathcal R,\mathcal C)$ is a grid minor and $R'_i \cap C'_j$ is horizontal, there is at least one row of 1 in each row-part of $(\mathcal R^*,\mathcal C^*)$.
  The corresponding $t$ vertices, together with exactly one vertex per column-part, form a biclique $K_{t,t}$ in $G$.

  $(ii) \Rightarrow (i)$.
  The contrapositive is immediate since a biclique $K_{t,t}$ yields a~$t$-grid minor in the adjacency matrix regardless of the vertex ordering.

$(ii) \Rightarrow (iii)$.  
Assume that there is an integer $d$ such that $\mathcal C$ is $d$-grid free.
Then by Marcus-Tardos theorem (\cref{thm:marcustardos}), there is a constant $c_d$ such that every graph of $\mathcal C$ has at most $c_dn/2$ edges.

$(iii) \Rightarrow (i)$.
We show the contrapositive, and the heredity of $\mathcal C$ is crucial here.
Observe that a hereditary class which is not $K_{t,t}$-free contains, for every integer $n$, a graph on $2n$ vertices with a (spanning) $K_{n,n}$.
Thus the average degree of the class is unbounded.

$(ii) \Rightarrow (iv)$.
If $\mathcal C$ is $d$-grid free, so is every subgraph of every $G \in \mathcal C$. 
Hence the subgraph closure $\sub(\mathcal C)$ of $\mathcal C$ also has bounded twin-width.

$(iv) \Rightarrow (i)$.
If $\mathcal C$ is \emph{not} $K_{t,t}$-free, $\sub(\mathcal C)$ contains every bipartite graph.
Thus $\sub(\mathcal C)$ has unbounded twin-width (for instance it contains the 1-subdivision of every clique).

$(v) \Rightarrow (i)$. 
  If $\mathcal C$ has expansion bounded by $f$, in particular $\nabla_0(\mathcal C) \leqslant f(0)$.
  Thus there exists $t := \lceil f(0) \rceil$ such that $\mathcal C$ is $K_{t,t}$-free.

At this point, we have shown that $(i), (ii), (iii), (iv)$ are all equivalent, and implied by~$(v)$.

$(i),(ii),(iii),(iv) \Rightarrow (v)$.   
Finally we assume that the first four conditions hold for a class~$\mathcal C$ of bounded twin-width.
Using all these assumptions, we want to bound the expansion of~$\mathcal C$.
The class $\mathcal B$ of binary structures obtained from $\sub(\mathcal C)$ by coloring the edge sets with two colors, in all possible ways, also has bounded twin-width.
Indeed it is $d$-grid free, so by \cref{thm:gridtheorem} it has bounded twin-width.

We first show that, for any fixed $r$, the class of $r$-shallow minors of $\mathcal C$ has bounded twin-width.
Indeed there is an FO transduction of $\mathcal B$ which contains all the $r$-shallow minors of $\mathcal C$, and we conclude by \cref{thm:transduction} (which works for graphs but more generally for binary structures with a constant number of binary relations).
To specify the transduction, we explain how every fixed $r$-shallow minor $H$ is obtained.
Let $G \in \mathcal B$ be an edge-bicolored graph containing $H$ as a spanning and induced $r$-shallow minor, where each contracted set induces a tree in $G$.
More precisely, the colors on $E(G)$ are such that every ``edge of $H$'' is colored~2, while every contracted edge (that is, other edge) is colored~1.
Let us recall that an FO transduction consists of adding a set of $O(1)$ non-deterministic unary relations (or coloring of the vertices with $O(1)$ colors), defining the new vertices and edges by means of FO formulas, and deleting all colors and potentially some vertices.
Here we only need one unary relation, say, $U$, and we focus on such a coloring where $U(v)$ holds for exactly one vertex $v$ in every contracted set.
The new vertices are simply defined by the formula $U(x)$.
Then we can define the edges by the formula $\phi(x,y) = U(x) \land U(y) \land \exists x' \exists y'~d^{2r}_1(x,x') \land d^{2r}_1(y,y') \land E_2(x',y')$, with $d^{2r}_1(z,z') = \bigvee_{i \in [2r]} \exists z_1 \cdots \exists z_i~z_1 = z \land z_i = z' \land \bigwedge_{j \in [i-1]} E_1(z_j,z_{j+1})$, where $E_1$ is the edge set colored 1, and $E_2$ is the edge set colored 2. 
The edge interpretation $\phi(x,y)$ links two vertices $u,v$ if they are ``reference vertices'' for their contracted set, and there is an edge colored 2 between two vertices $u', v'$ where there is a path of edges colored 1 of length at most $2r$ between $u$ and $u'$, and between $v$ and $v'$.
Such paths always exist within a contracted set since the radius is at most $r$, hence the diameter is at most $2r$.
Finally the graph obtained by this $(U,\phi)$-interpretation is exactly $H$.

We now want to bound the average degree of the $r$-shallow minors in $\nabla_r(\mathcal C)$ by some value $f(r)$.
Since $\nabla_r(\mathcal C)$ is subgraph-closed (every subgraph of an $r$-shallow minor is an $r$-shallow minor), $\sub(\nabla_r(\mathcal C))=\nabla_r(\mathcal C)$ has bounded twin-width.
Thus $(iv)$ implies $(iii)$ for the class $\nabla_r(\mathcal C)$.
Therefore $\nabla_r(\mathcal C)$ has bounded average degree, and $\mathcal C$ has bounded expansion.
\end{proof}

In the previous proof the heredity of $\mathcal C$ is only used to show that $(iii)$ implies $(i)$.
It is not an artifact of the proof since $\{K_{t,t} \uplus t^2K_1\}_{t \in \mathbb N}$ is a class of bounded twin-width where all graphs have linearly many edges, but admits arbitrary large bicliques.
The equivalences $(i) \Leftrightarrow (ii) \Leftrightarrow (iv) \Leftrightarrow (v)$ hold for every (possibly non-hereditary) class of bounded twin-width.
Bounded sparse twin-width classes remain surprisingly diverse.
They for instance contain $K_t$-minor free graphs and bounded-degree bounded twin-width graphs, which in turn contain some expander classes.
In particular bounded sparse twin-width graphs do not have polynomial expansion.

\subsection{Flat classes}\label{sec:flat}

For any graph invariant $\iota$, we say that a class $\mathcal C$ is \emph{$\iota$ flat} if it is included in $\sub(\mathcal G \boxtimes \mathcal H)$ with $\mathcal G$ and $\mathcal H$ two classes of bounded $\iota$, and $\mathcal H$ also has bounded degree.
Recalling the definition in \cref{sec:prelim}, a class is flat if it is treewidth flat.
We will see that twin-width flat classes have bounded twin-width.
It will imply that (treewidth) flat classes are other examples of bounded sparse twin-width classes.

We say that $G$ is a \emph{trigraph over a graph $H$} if $(V(G),E(G) \cup R(G))$ is isomorphic to $H$.
Thus $G$ is obtained from the graph $H$ by coloring red some of its edges.
More generally $G$ is a \emph{trigraph over a trigraph $H$} if there is a graph isomorphism from $(V(G),E(G) \cup R(G))$ to $(V(H),E(H) \cup R(H))$ such that every black edge of $G$ is mapped to a black edge of $H$.
Again $G$ is obtained from the trigraph $H$ by coloring red some of its black edges.
We start by bounding the twin-width of trigraphs over graphs with bounded degree and bounded twin-width.

\begin{lemma}\label{lem:degree-tww}
  Every trigraph over a graph $H$ has twin-width at most $\tww(H)+\Delta(H)$.
\end{lemma}
\begin{proof}
  Consider a $\tww(H)$-sequence of $H$.
  A simple but important observation is that the black degree of a vertex never increases in a contraction sequence.
  Thus each trigraph of the sequence has total degree at most~$\Delta(H)+\tww(H)$.
  Therefore, when the same sequence is applied to any trigraph over $H$, the overall maximum (red) degree is also bounded by~$\Delta(H)+\tww(H)$. 
\end{proof}

We can now show the following.

\product*
\begin{proof}
We set $d_G := \tww(G)$, $d_H := \tww(H)$, and $\Delta := \Delta(H)$, i.e., the maximum degree of $H$.
Let $G=G_n, \ldots, G_1=K_1$ be a sequence of $d_G$-contraction, and let $[h]$ be the vertex set of $H$, hence $h=|V(H)|$.
We set $d := \max\{(d_G+2)\Delta,d_H+\Delta\}$, and present a $d$-sequence for $G \boxtimes H$.
For a fixed $j \in [h]$, we call \emph{$j$-th copy of $G$}, the vertices $(v,j)$ of $G \boxtimes H$ for every $v \in V(G)$.

First we contract $G \boxtimes H$ to a trigraph over $H$ by a sequence containing as intermediate steps trigraphs over $G_n \boxtimes H, G_{n-1} \boxtimes H, \cdots, G_1 \boxtimes H$. 
Say $G_i$ is obtained from $G_{i+1}$ by contracting $u, v \in V(G_{i+1})$, into vertex $w$, then the part of the $d$-sequence from a trigraph over $G_{i+1} \boxtimes H$ to one over $G_i \boxtimes H$ consists of contracting, in any order, the vertices $(u,j)$ and $(v,j)$, into vertex $(w,j)$, for every $j \in [h]$. 
As the red degree of $w \in V(G_i)$ is at most~$d_G$, vertex $(w,j)$ has red degree at most $d_G(\Delta+1) + 2\Delta$. 
This is because the $j$-th copy of $G$ is linked to the $j'$-th copy only if $j' \in N_H[j]$.
This explains the $d_G(\Delta+1)$ term.
The additional $2\Delta$ accounts for possible red edges between $(w,j)$ and $(\star,j')$, where $\star \in \{u,v,w\}$ and $j' \neq j$. 

We can now finish the $d$-sequence from the obtained trigraph over $K_1 \boxtimes H$, which is isomorphic to $H$, using the $d_H$-sequence of $H$.
Indeed by \cref{lem:degree-tww} this trigraph admits a $d_H+\Delta$-sequence.
\end{proof}

\subproduct*
\begin{proof}
  By \cref{thm:product-stability}, $\mathcal G \boxtimes \mathcal H$ has twin-width bounded by a function of $\tww(\mathcal G)$, $\tww(\mathcal H)$, and $\Delta(\mathcal H)$.
  The implication $(i) \Rightarrow (iv)$ (via $(ii)$) in~\cref{thm:sparseboundedtww} does not require that the bounded twin-width class $\mathcal C$ is hereditary.
  Thus, $\mathcal G \boxtimes \mathcal H$ being $K_{t,t}$-free, the subgraph closure $\sub(\mathcal G \boxtimes \mathcal H)$ has twin-width bounded by a function of $\tww(\mathcal G)$, $\tww(\mathcal H)$, $\Delta(\mathcal H)$, and $t$.
\end{proof}

\begin{lemma}\label{lem:bicliquefreeSparse}
If $G$ is $K_{t,t}$-free, then $G\boxtimes H$ is $K_{s,s}$-free where $s=2t(\Delta(H)+1)$.
\end{lemma}
\begin{proof}
Assume, for the sake of contradiction, that there exist disjoint vertex sets $A, B \subseteq V(G\boxtimes H)$ such that $|A|=|B|=2t(\Delta(H)+1)$ and $A, B$ are fully adjacent. 
Let $V(H) := [h]$ and, let $a \in A$ be the vertex $(v,j)$ for some $v \in V(G)$ and $j \in V(H)$. 
Since $j$ is adjacent with at most $\Delta(H)$ vertices of $H$, and $(u,i)$ is adjacent with $(v,j)$ only if $i=j$ or $ij \in E(H)$, $B$ is contained in the union of at most $\Delta(H)+1$ copies of $G$.
This means that there exists some $j^* \in [h]$ such that the $j^*$-th copy of $G$ contains a set $B'$ of at least $2t$ vertices of $B$.
Likewise, there is an $i^* \in [h]$ such that the $i^*$-th copy of $G$ contains a set $A'$ of at least $2t$ of $A$. 
Let $A'' \subseteq A'$ and $B''\subseteq B'$ be vertex sets of size $t$ such that the first coordinates of the vertices in $A'' \cup B''$ are pairwise distinct.
Then the vertex subset of $V(G)$ which appears as the first coordinates in $A'' \cup B''$ form a $K_{t,t}$, a contradiction.
\end{proof}

\cref{thm:product-stability-class,lem:bicliquefreeSparse} imply that flat classes have bounded twin-width, since bounded treewidth classes have sparse bounded twin-width (they are $K_{t,t}$-free and have bounded twin-width).
In particular, it provides an alternative proof that planar graphs have bounded twin-width (see~\cite[Section 6]{twin-width1}).
The obtained bound remains bad since we still need to use \cref{thm:gridtheorem} to justify that the subgraph closure of a $K_{t,t}$-free bounded twin-width class has bounded twin-width.  

\subsection{Classes with bounded queue or stack number}\label{sec:bounded-queue-stack}

A pair of edges $uv$ and $xy$ is said \emph{independent} if $u,v,x,y$ are four distinct vertices.
An independent pair of edges $uv$ and $xy$ is \emph{nested} with respect to a linear ordering $\sigma$ of the vertex set, if $u \preceq_\sigma x \preceq_\sigma y \preceq_\sigma v$, and \emph{overlaps} if $u \preceq_\sigma x \preceq_\sigma v \preceq_\sigma y$.  
A \emph{queue} (resp.~\emph{stack}) \emph{layout} of a graph $G$ is a linear ordering $\sigma$ of $V(G)$ and a partition of $E(G)$ into $t$ parts, called queues (resp.~stacks), such that no independent pair of edges within the same part is nested (resp.~overlaps) with respect to $\sigma$.
The \emph{queue number} (resp.~\emph{stack number}) is defined as the minimum integer $t$ such that such a queue layout (resp.~stack layout) exists. 

\begin{lemma}\label{lem:quack}
  Let $\sigma$ be a linear ordering on the vertex set of a graph $G$.
  If $G$ admits an edge partition into $t$ parts such that each part forms a queue (resp.~a stack) with respect to $\sigma$, then the adjacency matrix $A_\sigma(G)$ is $2(t+1)$-grid free.
\end{lemma}
\begin{proof}
  Assume, for the sake of contradiction, that $A_{\sigma}(G)$ has a $2(t+1)$-grid minor $(\mathcal R, \mathcal C) := (\{R_1, \ldots, R_{2t+2}\}, \{C_1, \ldots, C_{2t+2}\})$, i.e., each zone $R_i \cap C_j$ contains an entry 1.
  Let us consider the $2 \times 2$ coarsening of $(\mathcal R, \mathcal C)$ where each part contains the first/last $t+1$ row/column parts of $(\mathcal R, \mathcal C)$.
  At least one of the two off-diagonal zones of this coarsening does \emph{not} cross the main diagonal of $A_{\sigma}(G)$.
  Without loss of generality, let us assume that all the vertices of $R_1, \ldots, R_{t+1}$ precedes all the vertices of $C_{t+2}, \ldots, C_{2t+2}$ in the order $\sigma$. 
  Now for each $i \in [t+1]$, choose one edge $u_iv_i$ from the zone $R_{t+2-i} \cap C_{t+1+i}$. 
  From the previous observation, we know that $u_i \preceq_{\sigma} v_i$ for each $i$.
  With respect to $\sigma$, the vertices $u_i$ (for $i$ going from 1 to $t+1$) form a decreasing sequence in $\sigma$ while the vertices $v_i$ form an increasing sequence. 
  Therefore the chosen $t+1$ edges are pairwise nested, contradicting that $G$ admits an edge partition into $t$ queues with respect to $\sigma$.

  For the stack number, we choose $t+1$ edges $u_iv_i$ from the zones $R_i \cap C_{t+1+i}$ for $i \in [t+1]$.
  They pairwise overlap, and thus contradict that there is a partition into $t$ stacks with that vertex ordering.
\end{proof}

The following is a direct consequence of \cref{thm:gridtheorem,lem:quack}.

\begin{theorem}\label{thm:quack}
  Classes with queue or stack number $t$ have twin-width bounded by $2^{2^{O(t)}}$.
\end{theorem}

\section{Twin-width of finitely generated groups}\label{sec:groups}

We investigate here an algebraic approach to constructing small graph classes.
Let $\Gamma$ be a (multiplicative) countable group where the identity is denoted by $\varepsilon$.
We assume that $\Gamma$ is generated by a finite set $S$.
We form the Cayley graph $\cay(\Gamma,S)$ which has vertex set $\Gamma$ and edge set all pairs $\{x,x \cdot s\}$ where $x \in \Gamma$ and $s \in S$. 

For example when $\Gamma$ is the free group generated by $S=\{a,b\}$, the graph $\cay(\Gamma,S)$ is the infinite tree where all the vertices have degree~4.
Furthermore, if we quotient $\Gamma$ by the relation $aba^{-1}b^{-1}=\varepsilon$, we obtain the infinite two-dimensional grid.
Both trees and grids are examples of classes with bounded twin-width.
Thus a natural question is whether this could hold for all finitely generated groups.
Let us denote by $F(\Gamma, S)$ the set of all finite induced subgraphs of $\cay(\Gamma,S)$.
Our main question in this section is the following.

\begin{conjecture}\label{conj:fgg}
For every group $\Gamma$ generated by a finite set $S$, the class $F(\Gamma, S)$ has bounded twin-width.
\end{conjecture}

This is a far-reaching generalization of the case of trees and grids.
It could provide some insights on both the structure of finite induced subgraphs of $\cay(\Gamma,S)$, but also in the global (infinite) structure of $\cay(\Gamma,S)$ as illustrated by the following result.

\begin{proposition}
If all the finite induced subgraphs of an infinite (possibly uncountable) graph $G$ have twin-width at most $t$, then there is a linear order $L$ on $V(G)$ such that the adjacency matrix of~$G$, ordered by~$L$, has no $f(t)$-mixed minor.
\end{proposition}

\begin{proof}
  Let $F(G)$ the class of finite non-empty induced subgraphs of~$G$.
  We assume that graphs in $F(G)$ have twin-width at most $t$, hence there exists an integer $f(t)$ such that any $H \in F(G)$ has a linear order~$L_H$ such that the adjacency matrix of $H$ has no $f(t)$-mixed minor.
  Let $\mathcal{M}_H$ be the logical structure formed by $H$ equipped with the order $L_H$.
  The proof proceeds in two parts.
  First, we will build an ultraproduct of $(\mathcal{M}_H)_{H \in F(G)}$, following a standard construction used, for example, to prove the compactness theorem.
  We will then show that $G$ is an induced subgraph of this ultraproduct.
  The order $L$ on $V(G)$ is then obtained by restriction of the order on the ultraproduct.

  For $H \in F(G)$, let $\uparrow\! H = \{H' \in F(G)\ |\ V(H) \subseteq V(H')\}$ be its upward closure in $F(G)$.
  The family of all $\uparrow\! H$ generates a proper filter on $F(G)$, which is contained in some ultrafilter~$U$.
  Let $\mathcal{M}' = \prod_{H \in F(G)} \mathcal{M}_H / U$ be the corresponding ultraproduct.
  By \L{}o\'s's theorem, any first-order formula satisfied by every $\mathcal M_H$ is also satisfied by $\mathcal M'$.
  Being a linear order, and being $f(t)$-mixed free with respect to that order can both be expressed in first-order logic, hence $\mathcal M'$ is an infinite graph equipped with a linear order for which it is $f(t)$-mixed free.

  Let us show that $G$ is an induced subgraph of $\mathcal M'$.
  For $v \in V(G)$, choose $\bar{v} \in \prod_{H \in F(G)} \mathcal{M}_H$ to be a tuple ``equal to $v$ when possible'', that is $\bar{v}(H) = v$ when $V(H) \ni v$ (and unconstrained if $V(H) \not\ni v$).
  We then map $v$ to the equivalence class of $\bar{v}$, which is a vertex in $\mathcal M'$.
  This mapping is injective: If $u \neq v$, then $\bar{u}(H) \neq \bar{v}(H)$ for any $H$ such that $u,v \in V(H)$, i.e., whenever $H \in \, \uparrow\!  G[\{u,v\}]$.
  Since $\uparrow\! G[\{u,v\}]$ is an element of $U$, by \L{}o\'s's theorem, $\bar{u}$ and $\bar{v}$ are not equated in $\mathcal M'$.
  The same arguments show that this mapping preserves edges and non-edges.
  Hence $G$ is an induced subgraph of $\mathcal M'$, and it follows that $G$ is $f(t)$-mixed free for the linear order on $\mathcal M'$ restricted to $G$.
\end{proof}

We suspect that bounded-degree Cayley graphs have bounded twin-width since they form a small class.
\begin{lemma}
The class $F(\Gamma, S)$ is small.
\end{lemma}

\begin{proof}
  Let us consider a finite induced subgraph $G$ of $F(\Gamma, S)$.
  We first assume that $G$ is connected.
  To describe $G$, it suffices to give a rooted spanning oriented tree $T$ in $G$ where each oriented edge $uv$ of $T$ is labeled by the generator $s$ in $S$ such that $u \cdot s=v$.
  To retrieve $G$ from $T$, one just has to fix the root of $T$ as $\varepsilon$ and deduce the set $V(G)$ of all the vertices by following the edges of $T$.
  The graph $G$ is then isomorphic to the subgraph of $\cay(\Gamma,S)$ induced by $V(G)$.
  Indeed, observe that this does not depend on the choice of $\varepsilon$, as any choice for mapping the root would be equivalent via multiplying to the left by some factor, which constitutes an automorphism of $\cay(\Gamma,S)$.

  If $G$ is not connected, we consider each connected component separately.
  In particular, the number of labeled graphs on vertex set $[n]$ that belong to the class $F(\Gamma, S)$ is at most the number of rooted forests whose edges are oriented and labeled by $|S|$ colors.
  By Cayley's formula there are $(n+1)^{n-1}$ labeled rooted forests on $n$ vertices~\cite{Cayley89}.
  Thus $F(\Gamma, S)$ has size at most~$(2|S|(n+1))^{n-1}$, hence this class is small.
\end{proof}

Should the small conjecture be true, $F(\Gamma, S)$ would have bounded twin-width.
We finally observe that having bounded twin-width is a group invariant, i.e., does not depend on the choice of the finite generating set $S$. 

\begin{lemma}
If $S$ and $S'$ are two finite generating sets of the group $\Gamma$, then $F(\Gamma, S)$ has bounded twin-width if and only if $F(\Gamma, S')$ has bounded twin-width.
\end{lemma}
\begin{proof}
  Let us assume that $F(\Gamma, S)$ has bounded twin-width.
  The first step is to show that a more general object has bounded twin-width.
  Namely, let us consider the oriented labeled Cayley graph $\OLCay(\Gamma,S)$ where every edge $\{x,x \cdot s\}$ is furthermore oriented from $x$ to $x \cdot s$ and labeled by $s$.
  Note that the class $\OLF(\Gamma, S)$ of all finite induced restrictions of $\OLCay(\Gamma,S)$ is contained is the (more general) class $\mathcal C$ of all orientations of graphs of $F(\Gamma, S)$ which are edge-labeled by $S$.
  The key fact is that $\mathcal C$ has bounded twin-width.
  Indeed, given any class of graphs $\mathcal G$ with degree at most~$d$ and twin-width at most $t$, the class ${\mathcal G}_s$ consisting of $\{1,\dots ,s\}$ edge-labeled orientations of graphs of $\mathcal G$ also has bounded twin-width.
  To see this, let us consider an element $O$ of ${\mathcal G}_s$ which is an oriented edge-labeled graph $G$ of ${\mathcal G}$.
  We just have to show that we can interpret $O$ in terms of $G$.
  To start with, we consider for $G$ a linear order $L_G$ of its vertices, such that the adjacency matrix of $G$, ordered by $L_G$, has twin-width at most~$f(t)$.
  When closed under induced restrictions, the class of birelations $(G,L_G)$ has bounded twin-width.
  Since the order $L_G$ provides for every vertex an order on its incident edges, we can furthermore label the vertices of $(G,L_G)$ using $2^d$ colors in order to code for every vertex $v$ how the (at most) $d$ edges incident to it are oriented.
  Therefore the class of orientations of $\mathcal G$ can be interpreted from the class of $(G,L_G)$ vertex-labeled by $2^d$ colors, and thus has bounded twin-width.
  For the edge-labeled version, we just have to label the vertices with $2^d|S|^d$ colors. 
  To conclude the proof, we observe that since every generator $s' \in S'$ can be expressed with $S$, the class $\OLF(\Gamma, S')$ is contained in an FO transduction of $\OLF(\Gamma, S)$.
  Therefore, by \cref{thm:transduction}, $\OLF(\Gamma, S')$, and thus $F(\Gamma, S')$, has bounded twin-width.
\end{proof}

Therefore, if the small conjecture does not hold, the class of finitely generated groups splits into bounded twin-width groups and unbounded twin-width groups.
This could reflect a known dichotomy for groups.
A natural candidate for a finitely generated group of unbounded twin-width, would be a group with no finite presentation.
For instance the lamplighter group is an interesting test case, but its associated class of graphs has indeed bounded twin-width.
A first step towards \cref{conj:fgg} is to show that finitely presented groups have bounded twin-width.


\begin{thebibliography}{10}

\bibitem{Alstrup17}
Stephen Alstrup, S{\o}ren Dahlgaard, and Mathias B{\ae}k~Tejs Knudsen.
\newblock Optimal induced universal graphs and adjacency labeling for trees.
\newblock {\em J. {ACM}}, 64(4):27:1--27:22, 2017.
\newblock URL: \url{https://doi.org/10.1145/3088513}, \href
  {http://dx.doi.org/10.1145/3088513} {\path{doi:10.1145/3088513}}.

\bibitem{Beineke69}
Lowell~W Beineke and Raymond~E Pippert.
\newblock The number of labeled k-dimensional trees.
\newblock {\em Journal of Combinatorial Theory}, 6(2):200--205, 1969.

\bibitem{Bilu06}
Yonatan Bilu and Nathan Linial.
\newblock Lifts, discrepancy and nearly optimal spectral gap*.
\newblock {\em Combinatorica}, 26(5):495--519, 2006.
\newblock URL: \url{https://doi.org/10.1007/s00493-006-0029-7}, \href
  {http://dx.doi.org/10.1007/s00493-006-0029-7}
  {\path{doi:10.1007/s00493-006-0029-7}}.

\bibitem{Blumensath10}
Achim Blumensath and Bruno Courcelle.
\newblock On the monadic second-order transduction hierarchy.
\newblock {\em Logical Methods in Computer Science}, 6(2), 2010.
\newblock URL: \url{https://doi.org/10.2168/LMCS-6(2:2)2010}, \href
  {http://dx.doi.org/10.2168/LMCS-6(2:2)2010}
  {\path{doi:10.2168/LMCS-6(2:2)2010}}.

\bibitem{twin-width3}
\'Edouard Bonnet, Colin Geniet, Eun~Jung Kim, Stéphan Thomassé, and Rémi
  Watrigant.
\newblock {Twin-width III: Max Independent Set and Coloring}.
\newblock {\em In preparation}, 2020.

\bibitem{twin-width1}
{\'{E}}douard Bonnet, Eun~Jung Kim, St{\'{e}}phan Thomass{\'{e}}, and
  R{\'{e}}mi Watrigant.
\newblock Twin-width {I:} tractable {FO} model checking.
\newblock {\em CoRR}, abs/2004.14789, 2020.
\newblock URL: \url{https://arxiv.org/abs/2004.14789}, \href
  {http://arxiv.org/abs/2004.14789} {\path{arXiv:2004.14789}}.

\bibitem{Cayley89}
Arthur Cayley.
\newblock A theorem on trees.
\newblock {\em Quart. J. Math.}, 23:376--378, 1889.

\bibitem{Cibulka16}
Josef Cibulka and Jan Kyncl.
\newblock F{\"{u}}redi-hajnal limits are typically subexponential.
\newblock {\em CoRR}, abs/1607.07491, 2016.
\newblock URL: \url{http://arxiv.org/abs/1607.07491}, \href
  {http://arxiv.org/abs/1607.07491} {\path{arXiv:1607.07491}}.

\bibitem{Dujmovic-als}
Vida Dujmovic, Louis Esperet, Gwena{\"{e}}l Joret, Cyril Gavoille, Piotr Micek,
  and Pat Morin.
\newblock Adjacency labelling for planar graphs (and beyond).
\newblock {\em CoRR}, abs/2003.04280, 2020.
\newblock URL: \url{https://arxiv.org/abs/2003.04280}, \href
  {http://arxiv.org/abs/2003.04280} {\path{arXiv:2003.04280}}.

\bibitem{Dujmovic-clustered}
Vida Dujmovi{\'c}, Louis Esperet, Pat Morin, Bartosz Walczak, and David~R Wood.
\newblock Clustered 3-colouring graphs of bounded degree.
\newblock {\em arXiv preprint arXiv:2002.11721}, 2020.

\bibitem{Dujmovic-queue}
Vida Dujmovic, Gwena{\"{e}}l Joret, Piotr Micek, Pat Morin, Torsten Ueckerdt,
  and David~R. Wood.
\newblock Planar graphs have bounded queue-number.
\newblock In David Zuckerman, editor, {\em 60th {IEEE} Annual Symposium on
  Foundations of Computer Science, {FOCS} 2019, Baltimore, Maryland, USA,
  November 9-12, 2019}, pages 862--875. {IEEE} Computer Society, 2019.
\newblock URL: \url{https://doi.org/10.1109/FOCS.2019.00056}, \href
  {http://dx.doi.org/10.1109/FOCS.2019.00056}
  {\path{doi:10.1109/FOCS.2019.00056}}.

\bibitem{Dujmovic-kplanar}
Vida Dujmovic, Pat Morin, and David~R. Wood.
\newblock The structure of k-planar graphs.
\newblock {\em CoRR}, abs/1907.05168, 2019.
\newblock URL: \url{http://arxiv.org/abs/1907.05168}, \href
  {http://arxiv.org/abs/1907.05168} {\path{arXiv:1907.05168}}.

\bibitem{DujmovicW05}
Vida Dujmovic and David~R. Wood.
\newblock Stacks, queues and tracks: Layouts of graph subdivisions.
\newblock {\em Discret. Math. Theor. Comput. Sci.}, 7(1):155--202, 2005.
\newblock URL: \url{http://dmtcs.episciences.org/346}.

\bibitem{Dvorak16}
Zdenek Dvor{\'{a}}k and Sergey Norin.
\newblock Strongly sublinear separators and polynomial expansion.
\newblock {\em {SIAM} J. Discrete Math.}, 30(2):1095--1101, 2016.
\newblock URL: \url{https://doi.org/10.1137/15M1017569}, \href
  {http://dx.doi.org/10.1137/15M1017569} {\path{doi:10.1137/15M1017569}}.

\bibitem{Dvorak10}
Zdenek Dvor{\'{a}}k and Serguei Norine.
\newblock Small graph classes and bounded expansion.
\newblock {\em J. Comb. Theory, Ser. {B}}, 100(2):171--175, 2010.
\newblock URL: \url{https://doi.org/10.1016/j.jctb.2009.06.001}, \href
  {http://dx.doi.org/10.1016/j.jctb.2009.06.001}
  {\path{doi:10.1016/j.jctb.2009.06.001}}.

\bibitem{Gavoille07}
Cyril Gavoille and Arnaud Labourel.
\newblock Shorter implicit representation for planar graphs and bounded
  treewidth graphs.
\newblock In Lars Arge, Michael Hoffmann, and Emo Welzl, editors, {\em
  Algorithms - {ESA} 2007, 15th Annual European Symposium, Eilat, Israel,
  October 8-10, 2007, Proceedings}, volume 4698 of {\em Lecture Notes in
  Computer Science}, pages 582--593. Springer, 2007.
\newblock URL: \url{https://doi.org/10.1007/978-3-540-75520-3\_52}, \href
  {http://dx.doi.org/10.1007/978-3-540-75520-3\_52}
  {\path{doi:10.1007/978-3-540-75520-3\_52}}.

\bibitem{Guillemot14}
Sylvain Guillemot and D{\'{a}}niel Marx.
\newblock Finding small patterns in permutations in linear time.
\newblock In {\em Proceedings of the Twenty-Fifth Annual {ACM-SIAM} Symposium
  on Discrete Algorithms, {SODA} 2014, Portland, Oregon, USA, January 5-7,
  2014}, pages 82--101, 2014.
\newblock URL: \url{https://doi.org/10.1137/1.9781611973402.7}, \href
  {http://dx.doi.org/10.1137/1.9781611973402.7}
  {\path{doi:10.1137/1.9781611973402.7}}.

\bibitem{Kannan92}
Sampath Kannan, Moni Naor, and Steven Rudich.
\newblock Implicit representation of graphs.
\newblock {\em {SIAM} J. Discret. Math.}, 5(4):596--603, 1992.
\newblock URL: \url{https://doi.org/10.1137/0405049}, \href
  {http://dx.doi.org/10.1137/0405049} {\path{doi:10.1137/0405049}}.

\bibitem{Klazar00}
Martin Klazar.
\newblock The f{\"u}redi-hajnal conjecture implies the stanley-wilf conjecture.
\newblock In {\em Formal power series and algebraic combinatorics}, pages
  250--255. Springer, 2000.

\bibitem{MarcusT04}
Adam Marcus and G{\'{a}}bor Tardos.
\newblock Excluded permutation matrices and the stanley-wilf conjecture.
\newblock {\em J. Comb. Theory, Ser. {A}}, 107(1):153--160, 2004.
\newblock URL: \url{https://doi.org/10.1016/j.jcta.2004.04.002}, \href
  {http://dx.doi.org/10.1016/j.jcta.2004.04.002}
  {\path{doi:10.1016/j.jcta.2004.04.002}}.

\bibitem{Marx05}
D{\'{a}}niel Marx.
\newblock Efficient approximation schemes for geometric problems?
\newblock In {\em Algorithms - {ESA} 2005, 13th Annual European Symposium,
  Palma de Mallorca, Spain, October 3-6, 2005, Proceedings}, pages 448--459,
  2005.
\newblock URL: \url{https://doi.org/10.1007/11561071\_41}, \href
  {http://dx.doi.org/10.1007/11561071\_41} {\path{doi:10.1007/11561071\_41}}.

\bibitem{MarxS13}
D{\'{a}}niel Marx and Ildik{\'{o}} Schlotter.
\newblock Cleaning interval graphs.
\newblock {\em Algorithmica}, 65(2):275--316, 2013.
\newblock URL: \url{https://doi.org/10.1007/s00453-011-9588-0}, \href
  {http://dx.doi.org/10.1007/s00453-011-9588-0}
  {\path{doi:10.1007/s00453-011-9588-0}}.

\bibitem{sparsity}
Jaroslav Nesetril and Patrice~Ossona de~Mendez.
\newblock {\em Sparsity - Graphs, Structures, and Algorithms}, volume~28 of
  {\em Algorithms and combinatorics}.
\newblock Springer, 2012.
\newblock URL: \url{https://doi.org/10.1007/978-3-642-27875-4}, \href
  {http://dx.doi.org/10.1007/978-3-642-27875-4}
  {\path{doi:10.1007/978-3-642-27875-4}}.

\bibitem{Norine06}
Serguei Norine, Paul~D. Seymour, Robin Thomas, and Paul Wollan.
\newblock Proper minor-closed families are small.
\newblock {\em J. Comb. Theory, Ser. {B}}, 96(5):754--757, 2006.
\newblock URL: \url{https://doi.org/10.1016/j.jctb.2006.01.006}, \href
  {http://dx.doi.org/10.1016/j.jctb.2006.01.006}
  {\path{doi:10.1016/j.jctb.2006.01.006}}.

\bibitem{PilipczukS19}
Michal Pilipczuk and Sebastian Siebertz.
\newblock Polynomial bounds for centered colorings on proper minor-closed graph
  classes.
\newblock In Timothy~M. Chan, editor, {\em Proceedings of the Thirtieth Annual
  {ACM-SIAM} Symposium on Discrete Algorithms, {SODA} 2019, San Diego,
  California, USA, January 6-9, 2019}, pages 1501--1520. {SIAM}, 2019.
\newblock URL: \url{https://doi.org/10.1137/1.9781611975482.91}, \href
  {http://dx.doi.org/10.1137/1.9781611975482.91}
  {\path{doi:10.1137/1.9781611975482.91}}.

\bibitem{Plotkin94}
Serge~A. Plotkin, Satish Rao, and Warren~D. Smith.
\newblock Shallow excluded minors and improved graph decompositions.
\newblock In {\em Proceedings of the Fifth Annual {ACM-SIAM} Symposium on
  Discrete Algorithms. 23-25 January 1994, Arlington, Virginia, {USA}}, pages
  462--470, 1994.
\newblock URL: \url{http://dl.acm.org/citation.cfm?id=314464.314625}.

\bibitem{Spinrad03}
Jeremy~P. Spinrad.
\newblock {\em Efficient graph representations}, volume~19 of {\em Fields
  Institute monographs}.
\newblock American Mathematical Society, 2003.
\newblock URL: \url{http://www.ams.org/bookstore-getitem/item=fim-19}.

\bibitem{Suk14}
Andrew Suk.
\newblock Coloring intersection graphs of x-monotone curves in the plane.
\newblock {\em Combinatorica}, 34(4):487--505, 2014.
\newblock URL: \url{https://doi.org/10.1007/s00493-014-2942-5}, \href
  {http://dx.doi.org/10.1007/s00493-014-2942-5}
  {\path{doi:10.1007/s00493-014-2942-5}}.

\end{thebibliography}

\end{document}